\documentclass[11pt, a4paper]{article}

\usepackage[english]{babel}
\usepackage[autostyle]{csquotes}
\usepackage[T1]{fontenc}
\usepackage{microtype}
\usepackage[shortlabels]{enumitem}

\usepackage{lmodern}
\usepackage{geometry}

\usepackage{customoperators}
\usepackage{bibliographysettings}

\usepackage[pdftex, bookmarksnumbered,colorlinks=false]{hyperref}
\usepackage{bookmark}

\usepackage[capitalize, nameinlink]{cleveref}
\crefformat{equation}{(#2#1#3)}
\crefrangeformat{equation}{(#3#1#4)--(#5#2#6)}
\crefmultiformat{equation}{(#2#1#3)}%
{ and~(#2#1#3)}{, (#2#1#3)}{, and~(#2#1#3)}

\theoremstyle{plain}
\newtheorem{thm}{Theorem}[section]
\newtheorem{prop}[thm]{Proposition}
\newtheorem{lem}[thm]{Lemma}

\newtheorem{assu}[thm]{Assumption}

\theoremstyle{definition}

\newtheorem*{rem}{Remark}

\counterwithin{thm}{section}

\AddToHook{env/lem/begin}{\crefalias{thm}{lem}}
\AddToHook{env/prop/begin}{\crefalias{thm}{prop}}
\AddToHook{env/cor/begin}{\crefalias{thm}{cor}}
\AddToHook{env/assu/begin}{\crefalias{thm}{assu}}
\AddToHook{env/defn/begin}{\crefalias{thm}{defn}}
\AddToHook{env/eg/begin}{\crefalias{thm}{eg}}
\AddToHook{env/rem/begin}{\crefalias{thm}{rem}}
\AddToHook{env/nota/begin}{\crefalias{thm}{nota}}

\crefname{thm}{Theorem}{Theorems}
\crefname{lem}{Lemma}{Lemmas}
\crefname{prop}{Proposition}{Propositions}
\crefname{cor}{Corollary}{Corollaries}
\crefname{assu}{Assumption}{Assumptions}
\crefname{defn}{Definition}{Definitions}
\crefname{eg}{Example}{Examples}
\crefname{rem}{Remark}{Remarks}
\crefname{nota}{Notation}{Notations}

\newcommand{\ao}{\mathfrak{a}}

\ifpdf
\hypersetup{
  pdfauthor={Lukas Reichmann, Arnaud Triay},
  pdftitle={On the spherical symmetry and finite-range assumptions of the interaction potential in the low energy study of dilute Bose gases},
}
\fi

\title{On the spherical symmetry and finite-range assumptions of the interaction potential in the low energy study of dilute Bose gases}
\author{Lukas Reichmann\footnote{Department of Mathematics, Friedrich-Alexander-Universität Erlangen-Nürnberg, Cauerstrasse 11, 91058 Erlangen, Germany} \and Arnaud Triay\footnote{Department of Mathematics, LMU Munich, Theresienstrasse 39, 80333 Munich, Germany}}
\date{\today}

\begin{document}
\maketitle
\pdfbookmark[section]{Abstract}{abstract} 
\begin{abstract}
We consider a Bose gas on the three-dimensional torus in the Gross--Pitaevskii regime and explain how to remove the assumptions of radiality and compact support on the interaction potential in the proof of Bose--Einstein condensation and convergence of the excitation spectrum to Bogoliubov's prediction. In particular, we sketch the proof of \cite{Reichmann_2025}.
\end{abstract}

%%%%%%%%%%%%%%%%%%%%%%%%%%%%%%%%%%%%%%%%%%%%%%%%%%%%%%%%
\section{Introduction}
The ideal Bose gas has been well understood since the pioneering works of Bose and Einstein in 1924 \cite{Bose_1924,Einstein_1924}. Einstein applied ideas used by Bose for photons to a gas of atoms and later predicted Bose--Einstein condensation (BEC) \cite{Einstein_1925}. 

The rigorous study of interacting Bose gases, on the other hand, remains an active area of research both in mathematics and in physics. In 1947, Bogoliubov introduced his theory for interacting Bose gases \cite{Bogoliubov_1947}, which allowed him to compute the low-lying excitation spectrum of the system and to justify Landau's criterion for superfluidity \cite{Landau_1941}. Namely, he showed that the dispersion relation of excited particles above the condensate behaves linearly at small momenta, a feature which is absent in ideal gases and illustrates the crucial role of interactions in the properties of dilute Bose gases.

The low-density Bose gas has been famously described by Lee, Huang, and Yang \cite{Lee_1957} and Wu \cite{Wu_1959}, who proposed a low-density expansion for its ground state energy
\begin{align*}
        e(\rho) = 4\pi \mathfrak{a} \rho^2 \bigg(1 + \frac{128}{15\sqrt{\pi}} \sqrt{\rho \mathfrak{a}^3} + 8 \left(\frac{4\pi}{3} - \sqrt{3}\right)\rho \mathfrak{a}^3 \log (\rho \mathfrak{a}^3) +\mathcal O(\rho \mathfrak{a}^3) \bigg),
\end{align*}
where $\rho$ is the gas density. A remarkable property of this expansion is its universality: it only depends on the interaction potential through its scattering length $\mathfrak{a}$, and not on the detailed shape of the potential, such as its radial symmetry or finite range. Indeed, the interaction between atoms in experiments arises from their electronic structure: it is usually approximated by a van der Waals attractive tail in $- \abs{x}^{-6}$ at large distances and by a strong repulsive core at short distances, preventing collapse, as for example in the Lennard--Jones potential, which behaves like $+ \abs{x}^{-12}$ near the origin. In practice, the effective interaction strength is even tuned using Feshbach resonances.  One therefore expects this expansion to remain valid for general short-range interaction potentials.

The mathematical justification of this expansion has attracted a lot of attention. Dyson showed, already in 1957, an upper bound matching the leading order \cite{Dyson_1957}, but the lower bound was only proved in 1998 by Lieb and Yngvason \cite{Lieb_1998}. It is worth noting that the mathematical treatment of potentials with signs is very delicate and remains a major open problem. While the case of weakly interacting systems \cite{Seiringer-11,GreSei-13,Lewin-15} can be handled, known results for dilute systems only cover potentials with a shallow negative part \cite{Lee-09,Yin-10b}. For this reason most works restrict to the case of purely repulsive interactions.

While the leading-order effect corresponds to two-body collisions, the second-order Lee--Huang--Yang correction is more subtle. It originates from the collective behavior of the gas and is described by the ground state energy of Bogoliubov's Hamiltonian governing the excitations above the condensate. Its mathematical justification was established by Yin and Yau \cite{Yau_2009}, who proved a matching upper bound using a trial state argument, while the matching lower bound was obtained by Fournais and Solovej \cite{Fournais_2020,Fournais_2023}. There have been several simplifications and extensions: in \cite{Basti_2021} a new proof of the upper bound was given, the formula has been extended to positive temperature in \cite{Seiringer-08,Yin-10,BasBocCenDeu-25,Haberberger_2024,Haberberger_2024upper}, the upper bound has been extended to hard-core potentials \cite{Basti_2026}, and even an upper bound capturing the next order has been proved \cite{Brooks_2025}.

A key observation underlying the analysis of dilute Bose gases is that the relevant low-energy excitations contributing to the Lee--Huang--Yang correction have momentum of order $\sqrt{\rho \mathfrak{a}}$, corresponding to the inverse of the healing length, also called the Gross--Pitaevskii length. Therefore, a large body of mathematical works has been dedicated to understanding the Bose gas at this length scale and have paved the way to a rigorous justification of the Lee--Huang--Yang formula.

A standard model is given by the so-called Gross--Pitaevskii regime, in which the particle number satisfies $N \simeq (\rho \mathfrak{a}^3)^{-1/2}$. In this ultra-dilute regime, the kinetic energy of the system is comparable to its interaction potential energy, which allows one to prove Bose--Einstein condensation and rigorously justify Bogoliubov's theory. For trapped systems, the condensate wave function is described by the Gross--Pitaevskii equation \cite{Gross_1961,Gross_1963,Pitaevskii_1961}, which was rigorously justified by Lieb, Seiringer, and Yngvason in 2000 \cite{Lieb_2000}. For weakly interacting systems, Bogoliubov's prediction was first rigorously justified by Seiringer \cite{Seiringer_2011} on the torus and extended to the trapped case \cite{GreSei-13}. The more singular Gross--Pitaevskii limit was later achieved by Boccato, Brennecke, Cenatiempo, and Schlein \cite{Boccato_2019}. In the Gross--Pitaevskii limit, the trapped case was independently solved by Brennecke, Schlein, and Schraven \cite{Brennecke_2022} and Nam and Triay \cite{Nam_2023}. See also \cite{Hainzl_2022,Brooks_2024} for simplifications of the proof, and \cite{Deuchert_2019,DeuSei-20,BocDeuSto-24,CapDeu-25} for studies of the system at positive temperature. The third-order contribution to the ground state energy of order $\log(N)/N$, corresponding to the Wu term, was derived in \cite{Caraci_2025}. 

A key ingredient in the derivation of Bogoliubov's theory is the Bose--Einstein condensation of low-energy states. It was first proved by Lieb and Seiringer \cite{Lieb_2002,Lieb_2006} for length scales slightly larger than the healing length, then improved with optimal rate in the Gross--Pitaevskii regime in \cite{Boccato_2020}. See also \cite{Nam_2016,Hainzl_2021,Fournais-21,AdhBreSch-21,BreBroCarOld-25,Junge_2026,Chong_2026} for other proofs. Surprisingly, all of these works assume the interaction potential to be radial and compactly supported. In fact, many later results build on BEC as an input, which is the only reason why they assume spherical symmetry of the interaction potential, see e.g.\ \cite{LamTri-25}.

An even stronger motivation to consider non-radial interactions comes from dipolar gases, where the atoms interact via the dipole-dipole interaction
\begin{align*}
K(x) = \frac{1-3\cos(\theta)^2}{\abs{x}^3},
\end{align*}
where $\cos(\theta)= x_3/\abs{x}$, see \cite{Chomaz-23} for a survey of experiments. However, due to its long range, the physics of these systems can no longer be solely described by s-wave scattering. The derivation of the Gross--Pitaevskii functional for such an interaction is still open, see \cite{Triay_2018} for partial results.

As already noted in \cite[Appendix C]{Lieb_2005}, and as done in \cite{BriSol-20,Fournais_2023}, the finite-range assumption can be easily lifted by an approximation argument if the interaction potential decays sufficiently fast.

On the other hand, spherical symmetry is often assumed from the start, already in the definition of the scattering length. Consider the zero-energy scattering equation \cite{Lieb_1998}
\begin{align*}
- 2 \Delta f + V f = 0, \qquad 0\leqslant f (x) \xrightarrow{\abs{x}\to \infty}1,
\end{align*}
which has a unique solution $f$ called the scattering function. When $V$ is radial and compactly supported, $f$ is also radial and since it is harmonic outside the support of $V$, the scattering length can be defined as the number $\ao$ such that
\begin{align*}
f(x) = 1 - \frac{\ao}{\abs{x}}, \quad x \notin \supp V.
\end{align*}
When $V$ is not radial but even, $V(x)= V(-x)$, we can still write
\begin{align*}
f(x) = 1 - \frac{1}{8\pi} \int_{\mathbb{R}^{3}} \frac{V(y)f(y)}{\abs{x-y}} \diff y = 1 - \frac{\ao}{\abs{x}} + \mathcal O\biggl(\frac{1}{\abs{x}^3}\biggr),
\end{align*}
where the dipole contribution in $\mathcal O(\abs{x}^{-2})$ vanishes by symmetry and where now the scattering length is given by 
\begin{align*}
\ao \coloneqq \frac{1}{8\pi} \int_{\mathbb{R}^{3}} V(y) f(y) \diff y.
\end{align*}
Of course, this is so far only a question of definition. %Let us move to the core of the problem. 
A key step in the proof of condensation in \cite{Lieb_2002,Lieb_2006} is the so-called Dyson lemma: assume $V$ is radial and compactly supported, then
\begin{align*}
-2 \Delta + V \geqslant 8\pi\ao U
\end{align*}
for any $0 \leqslant U$ radial with $\supp U \cap \supp V = \emptyset$ and $\int_{\mathbb{R}^{3}} U = 1$. Originally stated in \cite{Lieb_1998} to derive a lower bound for the leading order of the ground state energy, its proof relies crucially on the reduction to a one-dimensional problem. 

We will present a version of the Dyson lemma for non-radial potentials adapted from \cite{Ricaud_2023} (see \cref{lem:Dyson}), where the case of three-body interaction potentials (which are not radial) was studied. The lemma stated in \cite{Ricaud_2023} assumes the interaction potential to be bounded, and we extend it to integrable potentials. 

Let us mention other proofs of condensation that do not rely on the Dyson lemma. The general method is to implement a quadratic approximation of the Hamiltonian. If the interaction potential is sufficiently small, certain cubic and quartic error terms can directly be absorbed by the gap of the kinetic operator to prove condensation. This was done in \cite{Boccato_2018}. The technique was then simplified in \cite{Hainzl_2021} and allows for non-radial potentials. For larger potentials, these cubic and quartic error terms can also be renormalized via unitary transformations \cite{AdhBreSch-21}, or by completion of the square \cite{BreBroCarOld-25}. Another route is to use some localization argument and to propagate the gap of the kinetic operator to smaller subsystems where the cubic and quartic error terms are absorbable \cite{BriSol-20,Fournais_2020,Fournais_2023,Fournais-21,Chong_2026,Junge_2026}. Radiality seems to be always assumed to simplify computations and to use technical estimates on the scattering function, but it does not seem to be essential to the proof.

The article is organized as follows. In the next section, we state our main result \cref{thm:mainthm1}, which is the condensation of low-energy states and the derivation of the excitation spectrum without assuming either spherical symmetry or finite range of the potential. In \cref{sec:truncated_scattering_solution}, we study the scattering problem for non-radial potentials. In \cref{sec:gse}, we recover the first-order lower bound on the energy from the Dyson lemma. \cref{sec:fock} is dedicated to explaining some notations of the Fock space, while in \cref{ch:renorm}, we recall how to diagonalize the Hamiltonian and bound new error terms coming from the finite-range approximation. In \cref{ch:bec}, we explain how to prove optimal condensation and in \cref{ch:proof}, we complete the proof of the convergence of the excitation spectrum. As the proof is based on \cite{Reichmann_2025}, many parts of it, which are now standard, are sketched.

\section{Model and main result}

We consider a Bose gas on the torus in the Gross--Pitaevskii regime, i.e., the scaling of the interaction potential of the bosons is explicitly given by $N^2V(N\cdot)$.
While previous works on the low-energy excitation spectrum of Bose gases in this regime required radially symmetric interaction potentials with compact support, we relax these assumptions to the more physically realistic case of even potentials without compact support that decay sufficiently fast.

Note that similar assumptions on the interaction potential were made in \cite{Fournais_2023} for the proof of a lower bound on the thermodynamic limit of the ground state energy density. More precisely, their proof works under the same assumption \cref{eq:Vdecay} on the decay of the potential and even covers the case of hard-core interactions, but also requires the potential to be radially symmetric. 

Our proof requires approximating the interaction potential by a compactly supported potential and using properties of the scattering solution for non-radial potentials. 
Following the approach in \cite{Nam_2023} and using ideas from \cite{Haberberger_2024} and \cite{Hainzl_2022}, we extend the results obtained in \cite{Hainzl_2022,Boccato_2019} to repulsive, not necessarily compactly supported or radial $L^1$ potentials. 
In the following section, we present the precise setting and the main theorem we intend to prove. 

\subsection{Bose gases in a box}
We consider a Bose gas consisting of $N$ particles in a box $\Lambda=[-1/2,1/2]^3$ with periodic boundary conditions. 
The Hamiltonian acting on $L^2_\textnormal{s}(\Lambda^N)$, the space of square-integrable permutation-symmetric functions of $N$ variables $x_i\in\Lambda$, is given by
\begin{equation}\label{defn:Hamiltonianbox}
	H_N=\sum_{i=1}^N -\Delta_i + \sum_{1\leqslant i<j\leqslant N}V_N(x_i-x_j),
\end{equation}
where $V_N\in L^1(\Lambda)$ is the periodization of $\widetilde{V}_N\coloneqq N^2\widetilde{V}(N\cdot)$ for some non-negative function $\widetilde{V}$, i.e.,
\[
	V_N(x)=\sum_{z\in\Z^3}\widetilde{V}_N(x+z). 
\]
Since $H_N$ is non-negative, it can be defined as self-adjoint using Friedrichs' method. 
We denote the periodization of $\widetilde{V}$ by $V$ and we assume that $\widetilde{V}$ satisfies the following conditions. 
\begin{assu}\label{assu:V}
\begin{enumerate}
\item
$0\leqslant \widetilde{V}\in L^1(\R^3)$.
\item
$\widetilde{V}$ is even, i.e., $\widetilde{V}(x)=\widetilde{V}(-x)$.
\item\label{assu:Vc}
There exist constants $\gamma>1$ and $C$, such that 
for all $R>0$:
\begin{equation}\label{eq:Vdecay}
	\norm*{\widetilde{V}\ind_{B(0,R)^\mathsf{c}}}_{L^1}\leqslant CR^{-\gamma}.
\end{equation}
\end{enumerate}
\end{assu}
\begin{rem}
These conditions are satisfied, for example, if $0\leqslant\widetilde{V}\in L^1(\R^3)$ is even and decays at least like $\abs{x}^{-4-\varepsilon}$ for some $\varepsilon>0$. 
\end{rem}

We define the scattering length $\mathfrak{a}>0$ of $\widetilde{V}$ through
\begin{equation}\label{defn:scattering_length_1}
	8\pi\mathfrak{a}=\inf\Bigl\{\int_{\R^3}(2\abs{\gradient \varphi(x)}^2+\widetilde{V}(x)\abs{1-\varphi(x)}^2)\diff x:\varphi\in\dot{H}^1(\R^3)\Bigr\},
\end{equation}
where $\dot{H}^1(\R^3)\coloneqq\{\varphi\in L^1_\textnormal{loc}(\R^3):\norm{\gradient\varphi}_{L^2}<\infty \text{ and } \lim_{\abs{x}\to\infty}\varphi(x)=0\}$ is the homogeneous Sobolev space.
Here, $\lim_{\abs{x}\to\infty}\varphi(x)=0$ means that $\{x\in\R^3:\abs{\varphi(x)}>a\}$ has finite measure for all $a>0$. 

The main result we want to prove provides information about the low-energy spectrum of $H_N$ and the condensation of low-energy states. 
\begin{thm}\label{thm:mainthm1}
Let $H_N$ be defined as in \cref{defn:Hamiltonianbox} with $\widetilde{V}$ satisfying \cref{assu:V}. Denote by $\lambda_1(H_N)$ the ground state energy of $H_N$. Then, there exists $0<\kappa<1/13$ such that 
\[
	\lambda_1(H_N)=4\pi\mathfrak{a}(N-1)+e_\Lambda\mathfrak{a}^2
	+\frac12\sum_{p\in2\pi\Z^3\setminus\{0\}}\biggl(\sqrt{p^4+16\pi\mathfrak{a}p^2}-p^2-8\pi\mathfrak{a}+\frac{(8\pi\mathfrak{a})^2}{2p^2}\biggr)+\mathcal{O}(N^{-\kappa}),
\]
where
\[
	e_\Lambda=4\Biggl(\int_\Lambda\frac{1}{p^2}\diff p+\sum_{z\in\Z^3\setminus\{0\}}\int_\Lambda\biggl(\frac{1}{(p+z)^2}-\frac{1}{z^2}\biggr)\diff p\Biggr). 
\]
The spectrum of $H_N-\lambda_1(H_N)$ below a threshold $1\leqslant\Theta\leqslant N^{\kappa}$ consists of eigenvalues of the form 
\[
	\sum_{p\in 2\pi\Z^3\setminus\{0\}}n_p\sqrt{p^4+16\pi\mathfrak{a}p^2} + \mathcal{O}(N^{-\kappa}\Theta)
\]
with $n_p\in\N_0$ for all $p\in2\pi\Z^3\setminus\{0\}$, where only finitely many $n_p$ are non-zero.

Moreover, for every normalized sequence of approximate ground states $\Psi_N$ satisfying $\langle\Psi_N, H_N\Psi_N\rangle = \lambda_1(H_N) + \mathcal{O}(\widetilde{\Theta})$, $1\leqslant\widetilde{\Theta}\leqslant N$, we have
\begin{equation}\label{eq:BECrdm_theo}
	\frac{\langle u_0,\gamma_{\Psi_N}u_0\rangle}{N}= 1 + \mathcal O(N^{-1}\widetilde{\Theta}),
\end{equation}
where $u_0\equiv1$ and $\gamma_{\Psi_N}$ is the reduced density matrix
\begin{equation*}
	\gamma_{\Psi_N}(x;y)\coloneqq N\int_{\Lambda^{N-1}}\Psi_N(x,z_2,\dotsc,z_N)\overline{\Psi_N(y,z_2,\dotsc,z_N)}\diff z_2\dotso\diff z_N. 
\end{equation*}
\end{thm}

\begin{rem}
\begin{enumerate}
\item
Let $\gamma>1$ such that $\widetilde{V}$ satisfies \cref{eq:Vdecay} with this $\gamma$. 
\begin{itemize}
\item
If $\gamma\geqslant 14/11$, we can choose $0<\kappa<1/13$ arbitrarily in \cref{thm:mainthm1}. 
\item
If $\gamma\in(1,14/11)$, we can choose $0<\kappa<1/13$ such that 
\[
	\gamma=\frac{1+\kappa}{1-2\kappa}. 
\]
\item
The error for the eigenvalues of $H_N-\lambda_1(H_N)$ can be improved to be of order $\mathcal{O}(N^{-\kappa}\Theta^{1-\vartheta})$ for a threshold $1\leqslant\Theta\leqslant N^{\kappa/(1-\vartheta)}$, where $\kappa$ is chosen as just described. We can choose $0\leqslant\vartheta<1-12\kappa$ arbitrarily if $\gamma\geqslant6/5$ and
\[
	\vartheta=\frac{1-2\kappa}{3}
\]
if $\gamma\in(1,6/5)$.
\end{itemize}
\item
In \cite[equation (5.24)]{Boccato_2019}, the constant $e_\Lambda$ is given by the different formula
\[
	e_\Lambda=2-\lim_{M\to\infty}\sum_{\substack{p\in\Z^3\setminus\{0\} \\ \abs{p_1}, \abs{p_2}, \abs{p_3}\leqslant M}}\frac{4\cos(\abs{p})}{p^2}. 
\]
See \cite{Reichmann_2025} for the details on how to compute the formula for $e_\Lambda$ given in \cref{thm:mainthm1}.  
\item
The result in \cref{thm:mainthm1} was proved in detail in the master's thesis of one of the authors \cite{Reichmann_2025} under the assumptions that Bose--Einstein condensation holds and the leading order of the ground state energy is given by $4\pi\mathfrak{a}N$. In the following, we will explain the main steps of the proof and additionally justify Bose--Einstein condensation and the leading order of the ground state energy. 
\end{enumerate}
\end{rem}

It was proved in \cite{Lieb_2005} that, if $\widetilde{V}$ is \emph{radial} and decreases faster than $\abs{x}^{-3}$, the ground state energy $\lambda_1(H_N)$ of $H_N$ satisfies
\begin{equation}\label{eq:GSE}
	\lim_{N\to\infty}\frac{\lambda_1(H_N)}{N}= 4\pi\mathfrak{a}. 
\end{equation}
See also \cite{Lieb_1998} for an earlier proof in the case that $\widetilde{V}$ is \emph{radial} and \emph{compactly supported}. 
This means that the leading order of the ground state energy of $H_N$ depends only on the scattering length of $V$. 

In \cite{Lieb_2002} it was proved that, in the setting of \cite{Lieb_1998}, the ground state $\Psi_N$ exhibits complete Bose--Einstein condensation in the zero-momentum state $u_0\equiv1$, i.e.,
\begin{equation}\label{eq:BEC}
	\lim_{N\to\infty}\frac{\langle\Psi_N,a^\ast(u_0)a(u_0)\Psi_N\rangle}{N}=1,
\end{equation}
where $a^\ast$ and $a$ are the standard creation and annihilation operators on the bosonic Fock space. This even holds for every normalized sequence $\Psi_N$ of approximate ground states in the sense that $\langle\Psi_N,H_N\Psi_N\rangle=\lambda_1(H_N)+o(N)$.

\cref{thm:mainthm1} shows that for a large class of interaction potentials the ground state energy and the low-energy excitation spectrum are determined only by the scattering length of the interaction potential up to an error that vanishes in the limit $N\to\infty$.
We expect that this universal behaviour of the ground state energy persists in the third-order correction. In particular, we expect that the error estimates in \cref{thm:mainthm1} are not optimal and that the error term in the expansion of the ground state energy is actually of order $\mathcal{O}(\log(N)/N)$ as predicted by \cite{Wu_1959,Hugenholtz_1959,Sawada_1959}.

In comparison to our result, the interaction potential needed to be radial and compactly supported in \cite{Hainzl_2022,Boccato_2019,Caraci_2025} and additionally satisfy $\widetilde{V}\in L^2(\R^3)$ in \cite{Hainzl_2022} and $\widetilde{V}\in L^3(\R^3)$ in \cite{Boccato_2019,Caraci_2025}. 

Our proof can be easily adapted to extend similar results in the inhomogeneous setting to non-radial and non-compactly supported potentials. 

\subsection{Roadmap of the proof}
There are three main steps in the proof.

First, we extend the result from \cite{Ricaud_2023} on properties of the minimizer for the variational problem \cref{defn:scattering_length_1} for non-radial, compactly supported potentials in $L^\infty$ to non-radial, compactly supported potentials in $L^1$ in \cref{lem:scatteringsol} in \cref{sec:truncated_scattering_solution}. This requires adjusting the proof of the uniqueness of the minimizer as well as the proof of the pointwise bounds on the minimizer and its gradient. Additionally, we bound the difference of the scattering lengths of $\widetilde{V}$ and $\widetilde{V}\ind_{B(0,R)}$ for $R\geqslant1$ in \cref{lem:scatteringsol}. 

The second step is to justify the formula $4\pi\mathfrak{a}N$ for the leading order of the ground state energy of $H_N$ and BEC for approximate ground states. The lower bound for the leading order of the ground state energy will be proved in \cref{sec:gse}. The proof of the matching upper bound and BEC will be provided in \cref{ch:bec}. 

The last step is to use the method used in \cite{Nam_2023} to renormalize the Hamiltonian. 
We write 
\[
	H_N=H_{R,N}+H_\textnormal{err}
		=\biggl(\sum_{i=1}^N -\Delta_i + \sum_{1\leqslant i<j\leqslant N}V_{R,N}(x_i-x_j)\biggr)
			+\sum_{1\leqslant i<j\leqslant N}(V_N-V_{R,N})(x_i-x_j),
\]
where $V_{R,N}$ is the periodization of $N^2(\widetilde{V}\ind_{B(0,R)})(N\cdot)$. The parameter $R$ will be chosen to be $N$-dependent and to satisfy $1\ll R\ll N$.

The term $H_{R,N}$ can be renormalized as in \cite{Nam_2023}. First, $H_{R,N}$ is conjugated with the unitary operator $U$ introduced in \cite{Lewin_2015}, which extracts the excitations orthogonal to the condensate. We obtain $UH_{R,N}U^\ast=\ind_+^{\leqslant N}\mathcal{H}_{R,N}\ind_+^{\leqslant N}$, where $\mathcal{H}_{R,N}$ is an operator on the Fock space $\mathcal{F}$ and $\ind_+^{\leqslant N}$ is the projection onto the excitation Fock space $\mathcal{F}_+^{\leqslant N}$. In the next steps, $\mathcal{H}_{R,N}$ is conjugated with the unitary transforms $T_1,T_c$ as in \cite{Nam_2023}. These transforms are defined using the scattering solution of $\widetilde{V}\ind_{B(0,R)}$. Up to some error terms, we obtain a quadratic Hamiltonian, which only depends on the long-range potential $\varepsilon_{\ell,N}$, which we will define at the end of \cref{sec:truncated_scattering_solution}. This quadratic Hamiltonian can then be diagonalized using a standard Bogoliubov transform $T_2$. After justifying that we can replace $N\hat{\varepsilon}_{\ell,N}(p)$ with $8\pi\mathfrak{a}$, we will obtain
\[
	T^\ast_2T^\ast_cT^\ast_1\mathcal{H}_{R,N}T_1T_cT_2\approx E_N+\sq(E_\textnormal{Bog}),
\]
with
\begin{align*}
	\sq(E_\textnormal{Bog})&=\sum_{p\in\Lambda^\ast\setminus\{0\}}\sqrt{p^4+16\pi\mathfrak{a}p^2}a^\ast_pa_p, \\
	E_N&=4\pi\mathfrak{a}(N-1)+e_\Lambda\mathfrak{a}^2
	+\frac12\sum_{p\in\Lambda^\ast\setminus\{0\}}\biggl(\sqrt{p^4+16\pi\mathfrak{a}p^2}-p^2-8\pi\mathfrak{a}+\frac{(8\pi\mathfrak{a})^2}{2p^2}\biggr), \\
	e_\Lambda&=4\Biggl(\int_\Lambda\frac{1}{p^2}\diff p+\sum_{z\in\Z^3\setminus\{0\}}\int_\Lambda\biggl(\frac{1}{(p+z)^2}-\frac{1}{z^2}\biggr)\diff p\Biggr). 
\end{align*}

Conjugating $H_\textnormal{err}$ with $U$ yields four terms, which are error terms after conjugating with the unitary transforms $T_1,T_c,T_2$. More precisely, we have
\[
	UH_\textnormal{err}U^\ast
	=\ind_+^{\leqslant N}\mathcal{H}_\textnormal{err}\ind_+^{\leqslant N}
	\approx\ind_+^{\leqslant N}(H_{2,\textnormal{err}}+Q_{2,\textnormal{err}}+Q_{3,\textnormal{err}}+Q_{4,\textnormal{err}})\ind_+^{\leqslant N},
\]
with
\begin{align*}
H_{2,\textnormal{err}}&= N\int_{\Lambda^2}(V_N-V_{R,N})(x-y)a_x^\ast a_y\diff x\diff y, \\
Q_{2,\textnormal{err}}&=\frac{1}{2}\biggl(N-\numop-\frac12\biggr)\int_{\Lambda^2}(V_N-V_{R,N})(x-y)a_xa_y\diff x\diff y+\hc, \\
Q_{3,\textnormal{err}}&=\sqrt{(N-\numop)_+}\int_{\Lambda^2}(V_N-V_{R,N})(x-y)a^\ast_ya_xa_y\diff x\diff y+\hc, \\
Q_{4,\textnormal{err}}&=\frac12\int_{\Lambda^2}(V_N-V_{R,N})(x-y)a^\ast_x a^\ast_y a_xa_y\diff x\diff y. 
\end{align*}
The term $H_{2,\textnormal{err}}$ can be bounded immediately, while $Q_{3,\textnormal{err}}$ can only be bounded after conjugating with $T_1$ and $Q_{2,\textnormal{err}}$ can only be bounded after conjugating with $T_1$ and $T_c$. All bounds depend on $\norm{V_N-V_{R,N}}_{L^1}$. The bound on $T^\ast_cT^\ast_1Q_{2,\textnormal{err}}T_1T_c$ contains a term of order $\mathcal{O}(N\norm*{\widetilde{V}\ind_{B(0,R)^\mathsf{c}}}_{L^1})$, which is the reason for the assumption \cref{eq:Vdecay}. The term $Q_{4,\textnormal{err}}$ stays the same up to error terms after conjugating with $T_1,T_c,T_2$ and will be bounded in the proof of \cref{thm:mainthm1}. 

Finally, we will conclude the proof of \cref{thm:mainthm1} in \cref{ch:proof}. 

\subsection{Notation}
$C$ always denotes a positive constant that only depends on $\widetilde{V}$ and may vary from line to line. We will often omit the integration variables if there is no ambiguity. For $g$ defined on $\Lambda^2$, we will use the notation $g_x\coloneqq g(x,\cdot)$. We will use the following convention for the Fourier transform
\[
	\hat{g}(p)=\int g(x)e^{-\ii p\cdot x}\diff x,\quad\forall g\in L^1.
\]

%%%%%%%%%%%%%%%%%%%%%%%%%%%%%%%%%%%%%%%%%%%%%%%%%%%%%%%%
\section{Truncated scattering solution for non-radial, compactly supported potentials}\label{sec:truncated_scattering_solution}
It was proved, for example in \cite[Appendix~C]{Lieb_2005}, that \cref{defn:scattering_length_1} has a unique real-valued minimizer $0\leqslant \varphi\leqslant1$, which is spherically symmetric if $0\leqslant\widetilde{V}\in L^1(\R^3)$ is radial and compactly supported. 
A similar result for compactly supported $0\leqslant\widetilde{V}\in L^\infty(\R^d)$ with $d\geqslant3$, which is not necessarily radial, was obtained in \cite[Theorem~6]{Ricaud_2023}. We will adapt the proof strategy in \cite{Ricaud_2023} to prove the following lemma. 

\begin{lem}\label{lem:scatteringsol}
Let $0\leqslant \widetilde{V}\in L^1(\R^3)$ be even and denote $\widetilde{V}_R\coloneqq\widetilde{V}\ind_{B(0,R)}$ for $R\geqslant1$. 
The scattering length $\mathfrak{a}_R$ of $\widetilde{V}_R$ defined through
\begin{equation}\label{defn:scattering_length_2}
	8\pi\mathfrak{a}_R=\inf\Bigl\{\int_{\R^3}(2\abs{\gradient \varphi(x)}^2+\widetilde{V}_R(x)\abs{1-\varphi(x)}^2)\diff x:\varphi\in\dot{H}^1(\R^3)\Bigr\}
\end{equation}
satisfies 
\begin{equation}\label{eq:scattering_length_bound_1}
	8\pi\abs{\mathfrak{a}-\mathfrak{a}_R}\leqslant \norm*{\widetilde{V}-\widetilde{V}_R}_{L^1},
\end{equation}
with $\mathfrak{a}$ as in \cref{defn:scattering_length_1}. 
Moreover, \cref{defn:scattering_length_2} has a unique real-valued minimizer $\omega\in\dot{H}^1(\R^3)$, which is even, solves the scattering equation
\begin{equation}\label{eq:zero_scattering}
	-\Delta\omega=\frac12\widetilde{V}_R(1-\omega)
\end{equation}
in the sense of distributions, and satisfies the pointwise estimates
\[
	0\leqslant\omega(x)\leqslant1,\quad \omega(x)\leqslant\frac{C}{\abs{x}+1}
\]
for $x\in\R^3$, and 
\[
	\abs{\gradient\omega(x)}\leqslant\frac{C}{\abs{x}^2}
\]
for $\abs{x}\geqslant2R$, where $C$ is a constant independent of $R$. Finally, we have
\begin{equation}\label{eq:scat_len}
	8\pi\mathfrak{a}_R=\int_{\R^3}\widetilde{V}_R(1-\omega). 
\end{equation}
We suppress the $R$-dependence of $\omega$ to simplify the notation. 
\end{lem}

\begin{proof}
The existence of a real-valued minimizer $0\leqslant\omega\leqslant1$ and that it satisfies the scattering equation \cref{eq:zero_scattering} follow from the proof of \cite[Theorem~6]{Ricaud_2023}.

Consider the functionals
\begin{align*}
	\mathscr{E}[\varphi]&=\int_{\R^3}(2\abs{\gradient \varphi(x)}^2+\widetilde{V}_R(x)\abs{1-\varphi(x)}^2)\diff x, \\
	\mathscr{F}[\varphi]&=\int_{\R^3}(2\abs{\gradient \varphi(x)}^2+\widetilde{V}(x)\abs{1-\varphi(x)}^2)\diff x. 
\end{align*}
Due to the inequality $\widetilde{V}_R\leqslant\widetilde{V}$, we have $\mathscr{E}[\varphi]\leqslant\mathscr{F}[\varphi]$ for all $\varphi\in\dot{H}^1(\R^3)$. 
Thus, by the definitions \cref{defn:scattering_length_1} and \cref{defn:scattering_length_2} of $\mathfrak{a}$ and $\mathfrak{a}_R$, respectively, we have $\mathfrak{a}_R\leqslant\mathfrak{a}$. Hence, using $8\pi\mathfrak{a}_R=\mathscr{E}[\omega]$ and the inequality $8\pi\mathfrak{a}\leqslant\mathscr{F}[\omega]$, we obtain 
\[
	8\pi\abs{\mathfrak{a}-\mathfrak{a}_R}
	\leqslant\mathscr{F}[\omega]-\mathscr{E}[\omega]
	=\int_{\R^3}(\widetilde{V}-\widetilde{V}_R)(x)\abs{1-\omega(x)}^2\diff x
	\leqslant\norm*{\widetilde{V}-\widetilde{V}_R}_{L^1}. 
\] 

Now we can prove that $\omega$ is the unique real-valued minimizer of \cref{defn:scattering_length_2}. 
Since we have $0\leqslant\omega\leqslant1$ and $\widetilde{V}_R$ is non-negative, \cref{eq:zero_scattering} implies $\Delta\omega\leqslant0$, i.e., $\omega$ is superharmonic. 
Thus, by \cite[Theorem 9.4]{Lieb_2001}, we either have $\omega>0$ or $\omega=0$. 
	Let us first assume $\omega=0$. 
	By \cref{eq:zero_scattering} this can only be true if $\widetilde{V}_R=0$, in which case the variational problem \cref{defn:scattering_length_2} is trivial and $\omega = 0$ is the unique minimizer.
Let us now assume $\omega>0$. 
Assume there exists another real-valued minimizer $\widetilde{\omega}$ and define 
\[
h\coloneqq\frac{1}{\sqrt{2}}\sqrt{\omega^2+\widetilde{\omega}^2}. 
\]
Observe that $h\in\dot{H}^1(\R^3)$ due to the convexity inequality \cite[Theorem 7.8]{Lieb_2001} and the fact that $\omega,\widetilde{\omega}\in\dot{H}^1(\R^3)$. 
We compute
\begin{align*}
	8\pi\mathfrak{a}_R&\leqslant
	\mathscr{E}[h]
	=\int_{\R^3}(2\abs{\gradient h}^2+\widetilde{V}_R\abs{1-h}^2) \\
	&\leqslant\int_{\R^3}\bigl(\abs{\gradient \omega}^2+\abs*{\gradient\tilde{\omega}}^2
		+\widetilde{V}_R(1-\sqrt{2}\sqrt{\omega^2+\widetilde{\omega}^2}+\frac{1}{2}(\omega^2+\widetilde{\omega}^2))\bigr) \\
	&\leqslant\int_{\R^3}\bigl(\abs{\gradient \omega}^2+\abs*{\gradient\tilde{\omega}}^2
		+\widetilde{V}_R(1-\omega-\widetilde{\omega}+\frac{1}{2}(\omega^2+\widetilde{\omega}^2))\bigr)\\
	&=\frac12\bigl(\mathscr{E}[\omega]+\mathscr{E}[\tilde{\omega}]\bigr)
	=8\pi\mathfrak{a}_R,
\end{align*}
where we used the convexity inequality \cite[Theorem 7.8]{Lieb_2001} for the second inequality and the simple bound
\[
	x+y\leqslant\sqrt{2}\sqrt{x^2+y^2}
\]
for any $x,y\in\R$, for the third inequality. In particular, equality in the convexity inequality holds. Due to \cite[Theorem 7.8]{Lieb_2001} and the fact that $\omega>0$, this implies $\widetilde{\omega}=c\cdot\omega$ for some constant $c$.
Note that both $\omega$ and $\widetilde{\omega}$ solve the scattering equation \cref{eq:zero_scattering}. 
If we had $c\neq1$, this would imply $\widetilde{V}_R=0$, which is not possible, since in this case $\omega=0$ would be the unique minimizer of \cref{defn:scattering_length_2}. 
Therefore, $\omega$ is the unique real-valued minimizer of \cref{defn:scattering_length_2}. 

Since $\widetilde{V}_R$ is even, we observe that $\tilde{\omega}\coloneqq\omega(-\cdot)$ also minimizes \cref{defn:scattering_length_2}. Since $\omega$ is the unique real-valued minimizer, this implies that $\omega$ is an even function. 

Let 
\[
	G_y(x)\coloneqq \frac{1}{4\pi}\frac{1}{\abs{x-y}}
\]
be the Green's function.

Since $\omega$ is a non-negative superharmonic function, by \cite[Theorem~9.6]{Lieb_2001}, $\mu\coloneqq-\Delta\omega$ is a positive Borel measure on $\R^3$ and 
\begin{equation}\label{lem:scatteringsol:proof:omegaGreen}
	\tilde{\omega}(x)\coloneqq\int_{\R^3}G_y(x)\diff\mu(y)
\end{equation}
is finite for a.e.\ $x\in\R^3$ and satisfies $\omega=\tilde{\omega}+C$ a.e.\ for a constant $C\geqslant0$. Both $\omega$ and $\tilde{\omega}$ converge to $0$ as $\abs{x}\to\infty$ and are therefore equal. 
Since $\widetilde{V}_R$ is supported in $B(0,R)$, we have $\mu(A)=\mu(A\cap B(0,R))$ for every Borel measurable set $A$ and thus, for $\abs{x}\geqslant2R$ we obtain
\begin{align*}
	\omega(x)
	&=\frac{1}{4\pi}\int_{\R^3}\frac{1}{\abs{x-y}}\diff\mu(y)
	\leqslant\frac{1}{\abs{x}}\frac{1}{4\pi}\int_{B(0,R)}\frac{\abs{x}}{\abs{x}-\abs{y}}\diff\mu(y) \\
	&\leqslant \frac{1}{\abs{x}}\frac{1}{4\pi}\int_{B(0,R)}\bigl(\widetilde{V}_R(1-\omega)\bigr)(y)\diff y
	\leqslant\frac{1}{\abs{x}}\frac{1}{4\pi}\norm*{\widetilde{V}_R}_{L^1}\leqslant\frac{1}{\abs{x}}\frac{1}{4\pi}\norm*{\widetilde{V}}_{L^1},
\end{align*}
where we used $\widetilde{V}_R(1-\omega)\leqslant \widetilde{V}_R$. This implies for $x\in\R^3$
\[
	\omega(x)\leqslant\frac{C}{\abs{x}+1},
\]
since we have $0\leqslant\omega\leqslant1$. From \cref{lem:scatteringsol:proof:omegaGreen}, we obtain for $\abs{x}\geqslant2R$
\begin{equation}\label{eq:gradientomega}
	\gradient\omega(x)=-\frac{1}{8\pi}\int_{\R^3}\frac{\bigl(\widetilde{V}_R(1-\omega)\bigr)(y)}{\abs{x-y}^3}(x-y)\diff y. 
\end{equation}
Similarly to before, this implies 
\[
	\abs{\gradient\omega(x)}\leqslant\frac{C}{\abs{x}^2}
\]
for $\abs{x}\geqslant2R$. 

Let $0\leqslant\chi\leqslant1$ be smooth, radial, $\chi(x)=1$ if $\abs{x}\leqslant1/2$ and $\chi(x)=0$ if $\abs{x}\geqslant1$. Define $\omega_r\coloneqq\chi(\cdot/r)\omega$. It can be shown (see, e.g., \cite[proof of Lemma~2.81]{Frank_2022}) that $\omega_r\to\omega$ in $\dot{H}^1(\R^3)$. 
Since we have $\omega_r\in H^1(\R^3)$, we can find sequences $\{\varphi_n^{(r)}\}_{n\in\N}\subset C^\infty_c(\R^3)$ such that $\varphi_n^{(r)}\to\omega_r$ in $H^1(\R^3)$. 
Hence, we can rewrite the scattering length
\begin{align*}
	8\pi\mathfrak{a}_R
	&=\int_{\R^3}\bigl(2\abs{\gradient\omega}^2+\widetilde{V}_R\abs{1-\omega}^2\bigr) 
	=\lim_{r\to\infty}\int_{\R^3}\bigl(2\gradient\omega\cdot\gradient\omega_r+\widetilde{V}_R\abs{1-\omega}^2\bigr) \\
	&=\lim_{r,n\to\infty}\int_{\R^3}\bigl(2\gradient\omega\cdot\gradient\varphi_n^{(r)}+\widetilde{V}_R\abs{1-\omega}^2\bigr)
	=\lim_{r,n\to\infty}\int_{\R^3}\bigl(2\omega(-\Delta\varphi_n^{(r)})+\widetilde{V}_R\abs{1-\omega}^2\bigr) \\
	&=\lim_{r,n\to\infty}\int_{\R^3}\bigl(\widetilde{V}_R(1-\omega)\varphi_n^{(r)}+\widetilde{V}_R\abs{1-\omega}^2\bigr)
	=\lim_{r\to\infty}\int_{\R^3}\bigl(\widetilde{V}_R(1-\omega)\omega_r+\widetilde{V}_R\abs{1-\omega}^2\bigr) \\
	&=\int_{\R^3}\bigl(\widetilde{V}_R(1-\omega)\omega+\widetilde{V}_R\abs{1-\omega}^2\bigr)
	=\int_{\R^3}\widetilde{V}_R(1-\omega),
\end{align*}
where we used that $\omega$ minimizes \cref{defn:scattering_length_2} in the first equation
and \cref{eq:zero_scattering} in the fifth equation. 
The second to last equation follows by dominated convergence since $\omega_r\to\omega$ pointwise and $\abs*{\widetilde{V}_R(1-\omega)\omega_r}\leqslant \widetilde{V}_R\in L^1(\R^3)$. 
\end{proof}

Note that by \cref{eq:Vdecay} and \cref{eq:scattering_length_bound_1}, we have
\begin{equation}\label{eq:scattering_length_bound_2}
	\abs{\mathfrak{a}-\mathfrak{a}_R}\leqslant CR^{-\gamma}. 
\end{equation}
We will denote the periodization of $\widetilde{V}_R$ by $V_R$ and the periodization of $\widetilde{V}_{R,N}\coloneqq N^2\widetilde{V}_R(N\cdot)$ by $V_{R,N}$. 

As in \cite[Section~2.2]{Nam_2023}, we want to use a modified version of $\omega$ with a cut-off. For
\[
4RN^{-1}<\ell<1/2,
\] 
we define
\[
	\omega_{\ell,N}(x)=\chi(x/\ell)\omega(Nx),
\]
where $0\leqslant\chi\leqslant1$ is smooth, radial, $\chi(x)=1$ if $\abs{x}\leqslant1/2$ and $\chi(x)=0$ if $\abs{x}\geqslant1$. 
Using the notation $\omega_N=\omega(N\cdot)$ and $\chi_\ell=\chi(\cdot/\ell)$, we obtain
\begin{equation}\label{eq:truncated_zero_scattering}
	-2\Delta\omega_{\ell,N}=\widetilde{V}_{R,N}(1-\omega_N)-\varepsilon_{\ell,N}
\end{equation}
in the sense of distributions, with
\begin{equation}\label{defn:varepsilonellN}
	\varepsilon_{\ell,N}=2\Delta(\omega_{\ell,N}-\omega_N)=4\gradient\omega_N\cdot\gradient\chi_\ell+2\omega_N\Delta\chi_\ell,
\end{equation}
since $\Delta\omega_N\equiv0$ outside of $B(0,RN^{-1})\subset B(0,\ell/2)$.
Note that $\varepsilon_{\ell,N}\in C_c^\infty(\R^3)$, since $\varepsilon_{\ell,N}$ has support in $\{\ell/2\leqslant\abs{x}\leqslant\ell\}$ and $\omega_N$ is harmonic, and hence smooth, outside of $B(0,\ell/2)$. 
Both $\omega_{\ell,N}$ and $\varepsilon_{\ell,N}$ are even functions. 

%%%%%%%%%%%%%%%%%%%%%%%%%%%%%%%%%%%%%%%%%%%%%%%%%%%%%%%%
\section{Ground state energy: lower bound}\label{sec:gse}
We want to justify \cref{eq:GSE} for our chosen potential. In this section, we start by proving the lower bound 
\begin{equation}\label{eq:gselowerbound}
	\liminf_{N\to\infty}\frac{\lambda_1(H_N)}{N}\geqslant4\pi\mathfrak{a}. 
\end{equation}
We postpone the proof of the matching upper bound to the beginning of \cref{ch:bec}. Note that we will use neither \cref{eq:GSE} nor \cref{eq:BEC} until \cref{ch:bec}. 

By \cref{lem:scatteringsol}, we have $\mathfrak{a}_R\to\mathfrak{a}$ as $R\to\infty$. We also have the trivial lower bound
\[
	H_N=\sum_{i=1}^N -\Delta_i + \sum_{1\leqslant i<j\leqslant N}V_N(x_i-x_j)\geqslant\sum_{i=1}^N -\Delta_i + \sum_{1\leqslant i<j\leqslant N}V_{R,N}(x_i-x_j)\eqqcolon H_{R,N}.
\]
It suffices to prove the lower bound \cref{eq:gselowerbound} for $\widetilde{V}$ with compact support. Then we can deduce
\[
	\liminf_{N\to\infty}\frac{\lambda_1(H_N)}{N}\geqslant\liminf_{N\to\infty}\frac{\lambda_1(H_{R,N})}{N}\geqslant4\pi\mathfrak{a}_R\xrightarrow{R\to\infty}4\pi\mathfrak{a}. 
\]
In the remainder of this section, we assume that $\widetilde{V}$ is compactly supported with $\supp\widetilde{V}\subset B(0,R_0)$ for some fixed $R_0>0$. We follow the strategy used in \cite{Lieb_1998,Ricaud_2022}. First, we need to rescale $H_N$ to obtain an operator on $L^2_\textnormal{s}(\Lambda_N^N)$, where $\Lambda_N=[-N/2,N/2]^3$ is the box with periodic boundary conditions. Using the unitary rescaling transform $\mathcal{T}_{N}\Psi=N^{3N/2}\Psi(N\cdot)$, we have
\[
	\frac{1}{N^2}\mathcal{T}_{N}^\ast H_N\mathcal{T}_{N}=\sum_{i=1}^N -\Delta_i + \sum_{1\leqslant i<j\leqslant N}V_{\Lambda_N}(x_i-x_j)\eqqcolon H_{N,N},
\]
with
\[
	V_{\Lambda_N}(x)=\sum_{z\in\Z^3}\widetilde{V}(x+Nz). 
\]
Instead of the box $\Lambda_N=[-N/2,N/2]^3$, we will consider more generally the box $\Lambda_L=[-L/2,L/2]^3$ with periodic boundary conditions and prove the following result, which implies \cref{eq:gselowerbound}. 
\begin{prop}\label{prop:gselower}
Let $0\leqslant\widetilde{V}\in L^1(\R^3)$ be even, compactly supported with $\supp\widetilde{V}\subset B(0,R_0)$ for some $R_0>0$, and let $\mathfrak{a}$ be the scattering length of $\widetilde{V}$, defined as in \cref{defn:scattering_length_1}. Define 
\[
	H_{N,L}\coloneqq\sum_{i=1}^N -\Delta_i + \sum_{1\leqslant i<j\leqslant N}V_{\Lambda_L}(x_i-x_j)
\]
on $L^2_\textnormal{s}(\Lambda_L^N)$ and denote the ground state energy of $H_{N,L}$ by $\lambda_1(H_{N,L})$. Then, we have
\[
	\frac{\lambda_1(H_{N,L})}{N}\geqslant4\pi\rho\mathfrak{a}(1-\mathcal{O}(Y^{1/17}))
\] 
as $Y=\mathfrak{a}^3\rho\to0$, where $\rho=N/L^3$. 
\end{prop}

\begin{rem}
The bound in \cref{prop:gselower} holds uniformly in $\widetilde{V}$ if $R_0/\mathfrak{a}$ is bounded. 
\end{rem}
For the proof of \cref{prop:gselower}, we divide up $\Lambda_L$ into $M^3$ cubes of side length $\ell=\ell(L,N)>0$, where $M=M(L,N)\in\N$ and $M\ell=L$. We denote these cubes by $(B_i)_{1\leqslant i\leqslant M^3}$. Interactions between particles in different cubes are neglected, which lowers the ground state energy since $\widetilde{V}\geqslant0$. We impose Neumann boundary conditions on each cube, which is appropriate to find a lower bound on the ground state energy. The proof of \cref{prop:gselower} is a straightforward adaptation of the strategy used in \cite{Lieb_1998,Ricaud_2022} and we will only recall the main ingredients of the proof. 

We need a Dyson lemma such as \cite[Lemma~3]{Ricaud_2022} but for $L^1$ potentials instead of $L^\infty$ potentials. Using the properties of the scattering solution $\omega$ from \cref{lem:scatteringsol}, the proof of \cite[Lemma~3]{Ricaud_2022} shows the following lemma. 

\begin{lem}[Dyson lemma for non-radial potentials]\label{lem:Dyson}
Let $0\leqslant\widetilde{V}\in L^1(\R^3)$ be even, compactly supported with $\supp\widetilde{V}\subset B(0,R_0)$ for some $R_0>0$, and let $\mathfrak{a}$ be the scattering length of $\widetilde{V}$, defined as in \cref{defn:scattering_length_1}. Let $0\leqslant U\in L^1(\R^3)$ be radial with $\int_{\R^3}U=1$ and $\supp U\subset\{R_1\leqslant\abs{x}\leqslant R_2\}$. Let $B(0,R_2)\subset\Omega\subset\R^3$ be an open set. Then, we have on $L^2(\Omega)$
\[
	-2\gradient_x\ind_{B(0,R_2)}\gradient_x+\widetilde{V}(x)\geqslant 8\pi\mathfrak{a}\biggl(1-\frac{CR_0}{R_1}\biggr)U(x),
\]
where $C>0$ is a universal constant (independent of $\widetilde{V}$, $U$, $R_0$, $R_1$, $R_2$, $\Omega$). 
\end{lem}

The following lemma is a straightforward adaptation of \cite[Lemma~6]{Ricaud_2022} or \cite[Lemma~8]{Visconti_2026} to non-radial two-body interaction potentials. It follows directly from \cref{lem:Dyson} and the proof of \cite[Lemma~8]{Visconti_2026}. 
\begin{lem}[Many-body Dyson lemma]\label{lem:manybodyDyson}
Let $n\in\N$. Let $\widetilde{V}$ and $U$ be as in \cref{lem:Dyson}, with $\supp U\subset\{R/2\leqslant\abs{x}\leqslant R\}$. Define $\Lambda_{\ell,R}=[-\ell/2+R,\ell/2-R]^3$. 
Then, we have on $L^2_\textnormal{s}(\Lambda_\ell^n)$
\begin{multline*}
	\sum_{i=1}^n-\Delta_i+\frac{1}{2}\sum_{\substack{1\leqslant i,j\leqslant n \\ i\neq j}}\widetilde{V}(x_i-x_j) \\ 
	\geqslant4\pi\mathfrak{a}\biggl(1-\frac{CR_0}{R}\biggr)\sum_{\substack{1\leqslant i,j\leqslant n \\ i\neq j}}U(x_i-x_j)\ind_{\Lambda_{\ell,R}}(x_j)\prod_{\substack{1\leqslant k\leqslant n\\ k\neq i,j}}\ind_{\{\abs{x_j-x_k}\geqslant 2R\}}.
\end{multline*}
\end{lem}

To conclude the proof of \cref{prop:gselower}, we need a lower bound for the ground state energy of 
\[
	H_{n,\ell}\coloneqq\sum_{i=1}^n -\Delta_i + \sum_{1\leqslant i<j\leqslant n}\widetilde{V}(x_i-x_j)
\]
on $L^2_\textnormal{s}(\Lambda_\ell^n)$, where $\Lambda_\ell=[-\ell/2,\ell/2]^3$ is the box with Neumann boundary conditions and $n\leqslant4\rho\ell^3$. The proof of the following lemma is similar to that of \cite[Proposition~4]{Ricaud_2022} or \cite[Proposition~4]{Visconti_2026} and we skip the details. 

\begin{lem}\label{lem:gseHnllower}
Let $n\leqslant4\rho\ell^3$ and assume that $R_0/\mathfrak{a}$ is bounded. Then, for the choice $\ell=\mathfrak{a}Y^{-6/17}$ the ground state energy $E(n,\ell)$ of $H_{n,\ell}$ satisfies the lower bound 
\begin{equation}\label{eq:Enlbound}
	E(n,\ell)
	\geqslant\frac{4\pi\mathfrak{a}n(n-1)}{\ell^3}(1-\mathcal{O}(Y^{1/17}))
\end{equation}
when $Y\to0$. 
\end{lem}

\begin{rem}
We need to choose $\ell=\mathfrak{a}Y^{-6/17}\ll\mathfrak{a}Y^{-1/2}=(\rho\mathfrak{a})^{-1/2}=\ell_\textnormal{GP}$ much shorter than the Gross--Pitaevskii length scale, so we cannot apply the lower bound for the ground state energy of $H_{n,\ell}$ directly to $H_{N,N}$. 
\end{rem}

%%%%%%%%%%%%%%%%%%%%%%%%%%%%%%%%%%%%%%%%%%%%%%%%%%%%%%%%
\section{Preliminaries}\label{ch:preliminaries}
\subsection{Fock space formalism}\label{sec:fock}
In the following, let $\mathfrak{H}$ be a subspace of $L^2(\Lambda)$. 
We denote the bosonic Fock space by
\[
	\mathcal{F}(\mathfrak{H})\coloneqq\bigoplus_{n=0}^\infty\mathfrak{H}^n, 
	\qquad \mathfrak{H}^n\coloneqq\bigotimes_\textnormal{sym}^n\mathfrak{H},
	\qquad \mathfrak{H}^0\coloneqq\C. 
\]
For $g\in\mathfrak{H}$ we denote the standard creation and annihilation operators by $a^\ast(g)$ and $a(g)$, respectively. Additionally, we will use the operator-valued distributions $a^\ast_x$ and $a_x$ with $x\in\Lambda$ such that
\[
	a^\ast(g)=\int_\Lambda g(x)a^\ast_x\diff x,\quad a(g)=\int_\Lambda \overline{g(x)}a_x\diff x, \quad \forall g\in\mathfrak{H}. 
\]
These operators satisfy the standard canonical commutation relations. The second quantization of a self-adjoint operator $A$ on $\mathfrak{H}$ will be denoted by $\sq(A)$. 
In particular, the number operator $\mathcal{N}\coloneqq\sq(\id)$ can be rewritten as
\[
	\numop=\int_\Lambda a^\ast_xa_x\diff x=\sum_{n=1}^\infty a^\ast_na_n
\] 
for any orthonormal basis $\{u_n\}_{n\in\N}$ of $\mathfrak{H}$. 

\subsection{The excitation Hamiltonian}\label{sec:excitationham}

We denote $\mathcal{F}=\mathcal{F}(L^2(\Lambda))$. Let $\Lambda^\ast=2\pi\Z^3$. We can choose $\{u_p\}_{p\in\Lambda^\ast}$ with $u_p(x)=e^{\ii p\cdot x}$ explicitly as an orthonormal basis of $L^2(\Lambda)$ and we define $a^\sharp_p\coloneqq a^\sharp(u_p)$ with $\sharp\in\{\ast, \cdot\}$. Using the method of second quantization, we can write the Hamiltonian \cref{defn:Hamiltonianbox} as
\begin{equation}\label{defn:Hamiltonianboxsq}
\begin{split}
	H_N&=\sum_{p\in\Lambda^\ast}p^2a^\ast_pa_p+\frac12\sum_{p,q,r\in\Lambda^\ast}\widehat{V}_N(r)a^\ast_{p+r} a^\ast_qa_pa_{q+r}  \\
	&=\int_{\Lambda}a^\ast_x(-\Delta_x)a_x\diff x+\frac12\int_{\Lambda}\int_{\Lambda}V_N(x-y)a^\ast_xa^\ast_ya_xa_y\diff x\diff y. 
\end{split}
\end{equation}
The right-hand side of \cref{defn:Hamiltonianboxsq} is an operator on $\mathcal{F}$, but we will always consider its restriction to the $N$-body sector, which is equal to \cref{defn:Hamiltonianbox}. 
We will use the convention that the indices appearing in creation and annihilation operators are always non-zero. 
Let $Q=\id-\ketbra{u_0}{u_0}$ and
\[
L^2_+(\Lambda)=QL^2(\Lambda),\quad\mathcal{F}_+=\mathcal{F}(L^2_+(\Lambda)),\quad\mathcal{F}_+^{\leqslant N}=\bigoplus_{n=0}^N\bigl(L^2_+(\Lambda)\bigr)^n,\quad\ind_+^{\leqslant N}=\ind_{\mathcal{F}_+^{\leqslant N}}. 
\]
We will use the unitary transform $U\colon L^2_\textnormal{s}(\Lambda^N)\to\mathcal{F}_+^{\leqslant N}$ introduced in \cite{Lewin_2015}. It can be defined by
\[
	U(\Psi)=\bigoplus_{k=0}^N Q^{\otimes k}\biggl(\frac{a_0^{N-k}}{\sqrt{(N-k)!}}\Psi\biggr)
\]
for every $\Psi\in L^2_\textnormal{s}(\Lambda^N)$.
We have the following identities on $\mathcal{F}_+^{\leqslant N}$:\begin{equation}\label{eq:Uidentities}
\begin{split}
Ua^\ast_0a_0U^\ast&=N-\numop, \\
Ua^\ast(f)a_0U^\ast&=a^\ast(f)\sqrt{N-\numop}, \\
Ua^\ast_0a(f)U^\ast&=\sqrt{N-\numop}a(f), \\
Ua^\ast(f)a(g)U^\ast&=a^\ast(f)a(g)
\end{split}
\end{equation}
for all $f,g\in L^2_+(\Lambda)$. 
The proof of the following lemma is similar to that of \cite[Lemma~5]{Nam_2023} and we skip the details. 
\begin{lem}\label{lem:Ubound}
The following operator identity holds on $\mathcal{F}_+^{\leqslant N}$:
\[
	UH_NU^\ast=\ind_+^{\leqslant N}\mathcal{H}\ind_+^{\leqslant N},
\]
where 
\begin{equation}\label{defn:mathcalH}
\mathcal{H}=\frac{N-1}{2}\widehat{V}(0)+\sq(-\Delta)+H_2+Q_2+Q_3+Q_4+\mathcal{E}^{(U)}
\end{equation}
is an operator on the full Fock space $\mathcal{F}$ with 
\begin{align*}
H_2&=N\int_{\Lambda^2}V_{N}(x-y)a_x^\ast a_y\diff x\diff y, \\
Q_2&=\frac{1}{2}\biggl(N-\numop-\frac12\biggr)\int_{\Lambda^2}V_N(x-y)a_xa_y\diff x\diff y+\hc, \\
Q_3&=\sqrt{(N-\numop)_+}\int_{\Lambda^2}V_N(x-y)a^\ast_ya_xa_y\diff x\diff y+\hc, \\
Q_4&=\frac12\int_{\Lambda^2}V_N(x-y)a^\ast_x a^\ast_y a_xa_y\diff x\diff y,
\end{align*}
and $\mathcal{E}^{(U)}$ satisfies the quadratic form estimate on $\mathcal{F}$
\begin{equation}\label{error:U}
	\pm\mathcal{E}^{(U)}\leqslant C\frac{\numop^{2}}{N}+\varepsilon Q_4+\varepsilon^{-1}C\frac{(\numop+1)^2}{N^3}
\end{equation}
for all $\varepsilon>0$. 
\end{lem}
$\mathcal{H}$ is an operator on the full Fock space $\mathcal{F}$. In the following sections, we will restrict it to $\mathcal{F}_+$ to simplify the notation. 

\subsection{General quadratic transforms}\label{sec:general_quadratic_transforms}
It is well-known, see, e.g., \cite[Chapter~9]{Solovej_2014}, that for every real, symmetric Hilbert--Schmidt operator $s$ on $L^2(\Lambda)$, there exists a unitary transformation $T$ on the Fock space $\mathcal{F}$ such that 
\[
	T^\ast a^\ast(g)T=a^\ast(c(g))+a(s(\overline{g})),\quad\forall g\in L^2(\Lambda),
\]
where $c=\sqrt{1+s^2}$. 

We will need the following lemma about the conservation of the particle number.
\begin{lem}\label{quadratic:numop:bound}
The transformation $T$ satisfies 
\[
	T^\ast(\numop+1)^jT\leqslant C_j(1+\norm{s}_\textnormal{HS}^2)^j(\numop+1)^j,\quad\forall j\geqslant1,
\]
where $C_j$ is a constant independent of $s$. 
\end{lem}
The proof can be found, for example, in \cite[Lemma~4]{Nam_2023}. 

%%%%%%%%%%%%%%%%%%%%%%%%%%%%%%%%%%%%%%%%%%%%%%%%%%%%%%%%
\section{Renormalization}\label{ch:renorm}
\subsection{Definition of the first quadratic transform}\label{sec:q1}

Similarly to \cite[Section~3]{Haberberger_2024}, we define $s_1$, $\tilde{s}_1\colon\Lambda^2\to\R$ by
\begin{equation}\label{quadratic:defn:s}
	s_1=Q^{\otimes2}\tilde{s}_1,\quad \tilde{s}_1(x,y)=-\sum_{z\in\Z^3}N\omega_{\ell,N}(x-y+z),
\end{equation}
with $\omega_{\ell,N}$ defined as in \cref{sec:truncated_scattering_solution}. We will also denote the operator with integral kernel $s_1(x,y)$ by $s_1$. 

We want to use a quadratic transform similar to that in \cite[Section~4]{Nam_2023}. 
Since $s_1(x,y)$ can be shown to be in $L^2(\Lambda^2)$, $s_1$ is a Hilbert Schmidt operator on $L^2(\Lambda)$. Also, notice that $s_1(x,y)$ is real and symmetric. Hence, as remarked in \cref{sec:general_quadratic_transforms}, we can define the first quadratic transform $T_1$, that satisfies
\begin{equation}\label{eq:T1action}
	T_1^\ast a^\ast(g)T_1=a^\ast(c_1(g))+a(s_1(\overline{g})),\quad\forall g\in L^2(\Lambda). 
\end{equation}
This operator is a unitary transform on $\mathcal{F}_+$ since $s_1$ is an operator on $L^2_+(\Lambda)$. 

\subsection{Definition of the cubic transform}\label{sec:c}

We follow the strategy used in \cite[Section~5]{Haberberger_2024} and \cite[Section~5]{Nam_2023}. 
Let $\theta\in C^\infty(\R_{\geqslant0},[0,1])$ such that $\theta(x)=1$ for $x\leqslant1/2$ and $\theta(x)=0$ for $x\geqslant1$. We define $\theta_M=\theta(\numop/M)$ for $1\leqslant M\leqslant \ell N$. 

Recall $Q=\id-\ketbra{u_0}{u_0}$ and let $Q(x,y)$ be its integral kernel. We define
\[
	K_c^\ast=\frac{1}{\sqrt{N}}\int_\Lambda q^\ast_xa^\ast(s_{1,x})q_x\diff x,\qquad K_c=\frac{1}{\sqrt{N}}\int_\Lambda q^\ast_xa(s_{1,x})q_x\diff x,
\]
where 
\[
	q_x=\int_\Lambda Q(x,y)a_y\diff y=a(Q_x). 
\]
We use $q_x$ instead of $a_x$ in the definitions of $K_c^\ast$ and $K_c$ to ensure they leave $\mathcal{F}_+$ invariant. Since $q_x\ket{\xi}=a_x\ket{\xi}$ for all $\xi\in\mathcal{F}_+$, $q_x$ may be replaced by $a_x$ in all normal ordered expressions on $\mathcal{F}_+$.

We define the cubic transform 
\[
	T_c=\exp\biggl(\frac{\theta_M}{\sqrt{N}}\int_\Lambda q^\ast_xa^\ast(s_{1,x})q_x\diff x-\hc\biggr)=\exp(\theta_MK^\ast_c-K_c\theta_M). 
\]

We will need the following bound on the particle number. The proof is the same as the one of \cite[Lemma~19]{Nam_2023}. 

\begin{lem}\label{cubic:numop:bound}
Assume that $4RN^{-1}<\ell<1/2$, $M\leqslant N$, and $t\in[-1,1]$. On $\mathcal{F}$ we have
\[
	T_c^{-t}(\numop+1)^jT_c^t\leqslant C_j(\numop+1)^j,\quad\forall j\geqslant1,
\]
where $C_j$ is a constant independent of $N$. 
\end{lem}

We also need the following lemma, which is similar to \cite[Lemma~5.3]{Haberberger_2024} or \cite[Lemma~20]{Nam_2023}. 

\begin{lem}\label{cubic:lem:DeltaQ4}
Assume that $4RN^{-1}<\ell<1/2$, $MN^{-1}\leqslant\ell$, $NR^{-\gamma}\leqslant C$, and $t\in[0,1]$. We have on $\mathcal{F}_+$
\begin{align}
	T_c^{-t}\sq(-\Delta)T_c^{t}&\leqslant C(\sq(-\Delta)+Q_4+1), \label{cubic:lem:DeltaQ4:Deltabound} \\
	T_c^{-t}Q_4T_c^{t}&\leqslant C(Q_4+\numop+1+MN^{-1}\sq(-\Delta)). \label{cubic:lem:DeltaQ4:Q4bound} 
\end{align}
\end{lem}

\subsection{Definition of the second quadratic transform}\label{sec:q2}

In this section, we will define a transform similar to that in \cite[Section~6]{Haberberger_2024}. As remarked in \cref{sec:general_quadratic_transforms}, by choosing
\[
	s_2=\sum_{p\neq0}\sinh(\tau_p)\ketbra{u_p}{u_p},
\] 
we can define a quadratic transform on $\mathcal{F}_+$
\[
T_2=\exp\biggl(\frac12\sum_{p\neq0}\tau_p(a^\ast_pa^\ast_{-p}-a_pa_{-p})\biggr),
\] 
that satisfies 
\begin{equation}\label{eq:T2action}
	T^\ast_2a^\ast_pT_2=\cosh(\tau_p)a^\ast_p+\sinh(\tau_p)a_{-p},
\end{equation}
where 
\begin{align*}
	\tau_p&=\sinh^{-1}(\nu_p),\quad \nu_p=-\sgn(B_p)\sqrt{\frac12\Biggl(\frac{A_p}{\sqrt{A_p^2-B_p^2}}-1\Biggr)}, \\
	A_p&=p^2+N\hat{\varepsilon}_{\ell,N}(p),\quad B_p=N\hat{\varepsilon}_{\ell,N}(p),
\end{align*}
for $p\in\Lambda^\ast\setminus\{0\}$, since one can show that $\sum_p\abs{\nu_p}^2<\infty$, which is equivalent to $s_2$ being a Hilbert--Schmidt operator. 

We will need the following bounds on the number operator, the kinetic term $\sq(-\Delta)$, and the term $Q_4$. 
\begin{lem}\label{diagonal:lem:bounds}
Assume that $4RN^{-1}<\ell<1/2$, $MN^{-1}\leqslant\ell$, and $\ell$ is small enough. We have
\begin{align}
	T^{-t}_2(\numop+1)^jT^t_2&\leqslant C_j(\numop+1)^j, \label{diagonal:lem:bounds:numop} \\
	T^\ast_2\sq(-\Delta)T_2&\leqslant C(\sq(-\Delta)+\ell^{-1}), \label{diagonal:lem:bounds:kinetic} \\
	T^\ast_2Q_4T_2&\leqslant CQ_4+CN^{-1}\ell^{-2}+CN^{-1}(\numop+1)^2, \label{diagonal:lem:bounds:Q4} 
\end{align}
for all $j\geqslant1$ and $t\in[0,1]$. 
\end{lem}
\begin{proof}
The first bound \cref{diagonal:lem:bounds:numop} follows from \cref{quadratic:numop:bound}. 

The proof of the second bound \cref{diagonal:lem:bounds:kinetic} is similar to that of \cite[Lemma~6.2, (6.7)]{Haberberger_2024}. 

The proof of the third bound \cref{diagonal:lem:bounds:Q4} is similar to that of \cite[Lemma~18, (92)]{Hainzl_2022}. 
\end{proof}

\subsection{Bounding the new error terms}

Recall $\mathcal{H}=\frac{N-1}{2}\widehat{V}(0)+\sq(-\Delta)+H_2+Q_2+Q_3+Q_4+\mathcal{E}^{(U)}$ from \cref{defn:mathcalH} and write $\mathcal{H}=\mathcal{H}_{R,N}+\mathcal{H}_\textnormal{err}$ with $\mathcal{H}_\textnormal{err}=H_{2,\textnormal{err}}+Q_{2,\textnormal{err}}+Q_{3,\textnormal{err}}+Q_{4,\textnormal{err}}$, where
\begin{align*}
H_{2,\textnormal{err}}=H_2-H_{2,R}&= N\int_{\Lambda^2}(V_N-V_{R,N})(x-y)a_x^\ast a_y\diff x\diff y, \\
Q_{2,\textnormal{err}}=Q_2-Q_{2,R}&=\frac{1}{2}\biggl(N-\numop-\frac12\biggr)\int_{\Lambda^2}(V_N-V_{R,N})(x-y)a_xa_y\diff x\diff y+\hc, \\
Q_{3,\textnormal{err}}=Q_3-Q_{3,R}&=\sqrt{(N-\numop)_+}\int_{\Lambda^2}(V_N-V_{R,N})(x-y)a^\ast_ya_xa_y\diff x\diff y+\hc, \\
Q_{4,\textnormal{err}}=Q_4-Q_{4,R}&=\frac12\int_{\Lambda^2}(V_N-V_{R,N})(x-y)a^\ast_x a^\ast_y a_xa_y\diff x\diff y. 
\end{align*}
The term $\mathcal{H}_{R,N}$ can be dealt with as in \cite{Nam_2023,Haberberger_2024} using the transforms $T_1,T_c,T_2$ defined in \cref{sec:q1,,sec:c,,sec:q2} and properties of the scattering solution $\omega$ for the compactly supported potential $V_R$ from \cref{lem:scatteringsol}. Assume that $4RN^{-1}<\ell<1/2$, $MN^{-1}\leqslant\ell$, and $\ell$ is small enough. 
After conjugating with $T_1,T_c,T_2$, we obtain
\begin{equation}\label{eq:diagonal:RNbound}
	T^\ast_2T^\ast_cT^\ast_1\mathcal{H}_{R,N}T_1T_cT_2=E_N+\sq(E_\textnormal{Bog})+T^\ast_2Q_{4,R}T_2+\mathcal{E}_{R,N}
\end{equation}
on $\mathcal{F}_+$, with 
\begin{align*}
	\sq(E_\textnormal{Bog})&=\sum_{p\in\Lambda^\ast\setminus\{0\}}\sqrt{p^4+16\pi\mathfrak{a}p^2}a^\ast_pa_p, \\
	E_N&=4\pi\mathfrak{a}(N-1)+e_\Lambda\mathfrak{a}^2
	+\frac12\sum_{p\in\Lambda^\ast\setminus\{0\}}\biggl(\sqrt{p^4+16\pi\mathfrak{a}p^2}-p^2-8\pi\mathfrak{a}+\frac{(8\pi\mathfrak{a})^2}{2p^2}\biggr), \\
	e_\Lambda&=4\Biggl(\int_\Lambda\frac{1}{p^2}\diff p+\sum_{z\in\Z^3\setminus\{0\}}\int_\Lambda\biggl(\frac{1}{(p+z)^2}-\frac{1}{z^2}\biggr)\diff p\Biggr),
\end{align*}
and on $\mathcal{F}_+^{\leqslant N}$
\begin{equation*}
\begin{split}
	\pm\mathcal{E}_{R,N}
	&\leqslant CM^{1/2}N^{-1/2}\bigl(\sq(-\Delta)+T^\ast_2Q_4T_2+\ell^{-1}(\numop+1)\bigr) \\ 
	&\quad+C\ell^2\sq(-\Delta)+C\ell^4T^\ast_2Q_4T_2+C\ell^{1/2}(\numop+1)+C\ell^{1/2}\frac{(\numop+1)^{3/2}}{N^{1/2}} \\
	&\quad+\varepsilon(T^\ast_2Q_4T_2+\numop+1)+\varepsilon^{-1}C\biggl(\ell(\numop+1)+\frac{(\numop+1)^2}{N}+C_j\frac{(\numop+1)^{j+1}}{M^j}\biggr)
\end{split}
\end{equation*}
for all $j\geqslant1$ and $0<\varepsilon\leqslant1$.

The following lemma deals with the error term $\mathcal{H}_\textnormal{err}$. 

\begin{lem}\label{lem:newerror}
Assume that $4RN^{-1}<\ell<1/2$, $MN^{-1}\leqslant\ell$, and $\ell$ is small enough. On $\mathcal{F}_+^{\leqslant N}$ we have
\begin{align}
	\pm T^\ast_2T^\ast_cT^\ast_1H_{2,\textnormal{err}}T_1T_cT_2&\leqslant CR^{-\gamma}(\numop+1), \label{lem:newerror:H1} \\
	\begin{split}
	\pm T^\ast_2T^\ast_cT^\ast_1Q_{2,\textnormal{err}}T_1T_cT_2&\leqslant CNR^{-\gamma}+C\ell^{1/2}(\numop+1)+\frac{1}{4}T^\ast_2Q_4T_2 \\
	&\quad+\varepsilon(T^\ast_2Q_4T_2+\numop+1+MN^{-1}\sq(-\Delta)) \\ 
	&\quad+\varepsilon^{-1}C\biggl(\ell R^{-\gamma}(\numop+1)+\frac{(\numop+1)^2}{N}\biggr),
	\end{split}  \label{lem:newerror:Q2} \\
	\begin{split}
	\pm T^\ast_2T^\ast_cT^\ast_1Q_{3,\textnormal{err}}T_1T_cT_2&\leqslant C\ell^{1/2}\frac{(\numop+1)^{3/2}}{N^{1/2}} \\
	&\quad+\varepsilon (T^\ast_2Q_4T_2+\numop+1+MN^{-1}\sq(-\Delta)) \\
	&\quad+\varepsilon^{-1}C\biggl((R^{-\gamma}+\ell)(\numop+1)+\frac{(\numop+1)^2}{N}\biggr),
	\end{split}  \label{lem:newerror:Q3} \\
	T^\ast_2T^\ast_cT^\ast_1Q_{4,\textnormal{err}}T_1T_cT_2&=T^\ast_2Q_{4,\textnormal{err}}T_2+\mathcal{E}^{(Q_4)}, \label{lem:newerror:Q4}
\end{align}
with 
\begin{align*}
	\pm\mathcal{E}^{(Q_4)}&\leqslant CNR^{-\gamma}+C\ell^{1/2}(\numop+1)+CMN^{-1}(\sq(-\Delta)+T^\ast_2Q_4T_2+\numop+\ell^{-1}) \\
	&\quad+\varepsilon(T^\ast_2Q_4T_2+\numop+1)+\varepsilon^{-1}CR^{-\gamma}(\numop+1)
\end{align*}
for all $0<\varepsilon\leqslant1$.
\end{lem}

\begin{proof}
The first bound \cref{lem:newerror:H1} follows from the simple estimate
\[
	\pm H_{2,\textnormal{err}}\leqslant C\norm{V-V_R}_{L^1}\numop\leqslant CR^{-\gamma}\numop
\]
and \cref{quadratic:numop:bound,cubic:numop:bound}. 

As in the proof of \cite[Lemma~10]{Nam_2023}, we obtain
\begin{equation}\label{expansionT1Q2witherror}
\begin{split}
T^\ast_1Q_{2,\textnormal{err}}T_1&=\frac{1}{2}\int_{\Lambda^2}N(V_N-V_{R,N})(x-y)(a^\ast_xa^\ast_y+a_xa_y) \\
	&\quad+(N-T^\ast_1\numop T_1-1/2)\int_{\Lambda^2}(V_N-V_{R,N})(x-y)s_1(x,y)+\mathcal{E}_1,
\end{split}
\end{equation}
with
\[
	\pm\mathcal{E}_1\leqslant C\ell^{1/2}(\numop+1)+C\ell^{1/2}N^{-1}(\numop+1)^2+\varepsilon Q_4+\varepsilon^{-1}CN^{-1}(\numop+1)^2,\quad\forall\varepsilon>0. 
\]
To bound the second term on the right-hand side of \cref{expansionT1Q2witherror}, we use $\norm{s_1}_{L^\infty}\leqslant CN$, the bound $\pm(N-T^\ast_1\numop T_1-1/2)\leqslant CN$ on $\mathcal{F}_+^{\leqslant N}$, and $\norm{V-V_R}_{L^1}\leqslant CR^{-\gamma}$. Hence, we have
\[
	T^\ast_1Q_{2,\textnormal{err}}T_1=\frac{1}{2}\int_{\Lambda^2}N(V_N-V_{R,N})(x-y)(a^\ast_xa^\ast_y+a_xa_y)
	+\mathcal{E}_2\eqqcolon Q_{2,\textnormal{err}}^{(T_1)}+\mathcal{E}_2,
\]
with 
\[
	\pm\mathcal{E}_2\leqslant CNR^{-\gamma}+C\ell^{1/2}(\numop+1)+C\ell^{1/2}N^{-1}(\numop+1)^2+\varepsilon Q_4+\varepsilon^{-1}CN^{-1}(\numop+1)^2,\quad\forall\varepsilon>0. 
\]
We can bound $T^\ast_2T^\ast_c\mathcal{E}_2T_cT_2$ using \cref{quadratic:numop:bound,,cubic:numop:bound,,cubic:lem:DeltaQ4,,diagonal:lem:bounds}.

By the Duhamel formula, we have
\[
	T^\ast_cQ_{2,\textnormal{err}}^{(T_1)}T_c=Q_{2,\textnormal{err}}^{(T_1)}
	+\int_0^1T_c^{-t}[Q_{2,\textnormal{err}}^{(T_1)},\theta_MK^\ast_c-K_c\theta_M]T_c^t\diff t. 
\]
Using $V_N-V_{R,N}\leqslant V_N$, we obtain
\[
	\pm Q_{2,\textnormal{err}}^{(T_1)}
	\leqslant \frac18Q_4+C\norm*{N^2(V_N-V_{R,N})}_{L^1}\leqslant \frac18Q_4+CNR^{-\gamma}. 
\]
Similarly to the proofs of \cite[Lemma~5.8]{Haberberger_2024} and \cite[Lemma~25]{Nam_2023}, we can bound 
\begin{equation*}
	\pm[Q_{2,\textnormal{err}}^{(T_1)},\theta_MK^\ast_c-K_c\theta_M]
	\leqslant CM^{1/2}N^{1/2}R^{-\gamma}+\frac{\varepsilon}{8}Q_4+\varepsilon^{-1}C\ell R^{-\gamma}(\numop+1)
\end{equation*}
for all $0<\varepsilon\leqslant1$. 
Now we can use \cref{cubic:numop:bound,,cubic:lem:DeltaQ4,,diagonal:lem:bounds} to conclude the proof of \cref{lem:newerror:Q2}. 

The bound \cref{lem:newerror:Q3} can be obtained similarly using the analysis in \cite{Nam_2023}. This bound is simpler because it only requires computing $T^\ast_1Q_{3,\textnormal{err}}T_1$ and we skip the details. 

For the proof of \cref{lem:newerror:Q4} we can proceed as for $Q_{4,R}$, which is similar to the analysis in \cite{Nam_2023,Haberberger_2024}. The quadratic term appearing after computing $T^\ast_1Q_{4,\textnormal{err}}T_1$ can be dealt with exactly as $Q_{2,\textnormal{err}}^{(T_1)}$. We skip the details of the proof.  
\end{proof}

Combining \cref{eq:diagonal:RNbound} with \cref{lem:newerror}, we obtain 
\begin{equation}\label{eq:diagonal:Hbound}
	T^\ast_2T^\ast_cT^\ast_1\mathcal{H}T_1T_cT_2=E_N+\sq(E_\textnormal{Bog})+T^\ast_2Q_{4}T_2+\mathcal{E}^{(T_2)},
\end{equation}
where $\mathcal{E}^{(T_2)}$ satisfies on $\mathcal{F}_+^{\leqslant N}$
\begin{equation}\label{eq:diagonal:lem:main:errorestimate}
\begin{split}
	\pm\mathcal{E}^{(T_2)}
	&\leqslant CNR^{-\gamma}+CM^{1/2}N^{-1/2}\bigl(\sq(-\Delta)+T^\ast_2Q_4T_2+\ell^{-1}(\numop+1)\bigr) \\
	&\quad+C\ell^2\sq(-\Delta)+C\ell^4T^\ast_2Q_4T_2+C\ell^{1/2}(\numop+1)+C\ell^{1/2}\frac{(\numop+1)^{3/2}}{N^{1/2}} \\
	&\quad+\frac14T^\ast_2Q_4T_2+\varepsilon(T^\ast_2Q_4T_2+\numop+1) \\
	&\quad+\varepsilon^{-1}C\biggl((R^{-\gamma}+\ell)(\numop+1)+\frac{(\numop+1)^2}{N}+C_j\frac{(\numop+1)^{j+1}}{M^j}\biggr)
\end{split}
\end{equation}
for all $j\geqslant1$ and $0<\varepsilon\leqslant1$.

%%%%%%%%%%%%%%%%%%%%%%%%%%%%%%%%%%%%%%%%%%%%%%%%%%%%%%%%
\section{Optimal BEC}\label{ch:bec}

\subsection{Ground state energy: upper bound}\label{sec:gseup}
In this section, we want to finish the proof of \cref{eq:GSE}. We proved the lower bound \cref{eq:gselowerbound} in \cref{sec:gse} and it remains to prove the corresponding upper bound
\begin{equation}\label{eq:gseupperbound}
	\limsup_{N\to\infty}\frac{\lambda_1(H_N)}{N}\leqslant4\pi\mathfrak{a}. 
\end{equation}
Recall that $UH_NU^\ast=\ind_+^{\leqslant N}\mathcal{H}\ind_+^{\leqslant N}$, and by \cref{eq:diagonal:Hbound}, we have
\[
	T^\ast_2T^\ast_cT^\ast_1\mathcal{H}T_1T_cT_2=E_N+\sq(E_\textnormal{Bog})+T^\ast_2Q_4T_2+\mathcal{E}^{(T_2)}
\]
on $\mathcal{F}_+$. Let $\Omega=(1,0,0,\dotsc)\in\mathcal{F}_+$ be the vacuum. Using \cref{quadratic:numop:bound,cubic:numop:bound}, we obtain for all $j\geqslant1$
\begin{align}
	\langle T_1T_cT_2\Omega,\numop^j T_1T_cT_2\Omega\rangle
	\leqslant C \langle\Omega,(\numop+1)^j\Omega\rangle
	\leqslant C, 
	\label{eq:est_Nj_omega}
\end{align}
where, of course, the constants $C$ depend on $j$.
Note that $T_1T_cT_2$ is a unitary operator on $\mathcal{F}_+$, so we need to consider $\ind_+^{\leqslant N}T_1T_cT_2\Omega \in \mathcal{F}_+^{\leqslant N}$, which satisfies
\begin{equation*}
\| \ind_+^{\leqslant N}T_1T_cT_2\Omega \|^2 = 1 + \mathcal O(N^{-j})
\end{equation*}
for all $j\geqslant1$. This can be deduced easily from \cref{eq:est_Nj_omega}. Now we estimate the off-diagonal part of $\mathcal H$. Using that $\mathcal H \geqslant0$, we find on $\mathcal F^{\leqslant N}_+$
\begin{align}
\mathds{1}^{\leqslant N}_+ \mathcal H \mathds{1}^{\leqslant N}_+ &- \mathcal H \nonumber  \\
	&\leqslant - \mathds{1}^{\leqslant N}_+ \mathcal H \mathds{1}^{> N}_+ + \hc \nonumber  \\
	&= - \frac{1}{2} \mathds{1}^{\leqslant N}_+ \left( \sqrt{((N-\mathcal N)(N-1-\mathcal N))_+}\int_{\Lambda^2}V_N(x-y)a_xa_y \right)\mathds{1}^{> N}_+ + \hc  \nonumber \\
	&\quad - \mathds{1}^{\leqslant N}_+ \sqrt{(N-\numop)_+}\int_{\Lambda^2}V_N(x-y)a^\ast_ya_xa_y\diff x\diff y) \mathds{1}^{>N}_+ + \hc  \nonumber \\
	&\leqslant \eta Q_4 + \eta^{-1} C N \widehat{V}(0) \frac{\mathcal N^{j}}{N^{j}} \label{eq:off_diag_HN},
\end{align}
where we denoted $\mathds{1}^{\leqslant N}_+ = \mathds{1}_{\mathcal F^{\leqslant N}_+} = \id - \mathds{1}^{> N}_+$ and used the Chebyshev inequality $\mathds{1}^{>N}_+ \leqslant \mathcal N^j/N^j$. 

Let us gather some bounds on the expectation of $Q_4$ and $\mathcal{E}^{(T_2)}$ before concluding. Using \cref{cubic:lem:DeltaQ4,diagonal:lem:bounds} and \cite[Lemma 3.10]{Reichmann_2025}, we find the simple bound
\begin{align*}
\langle T_1T_cT_2\Omega, Q_4 T_1T_cT_2\Omega\rangle \leqslant C N
\end{align*}
and from \cref{diagonal:lem:bounds} only, we obtain 
\[
	\langle T_2\Omega, Q_4T_2\Omega\rangle
	\leqslant CN^{-1}\ell^{-2}+CN^{-1}.
\]
From \cref{eq:diagonal:lem:main:errorestimate}, we obtain
\begin{align*}
	\langle\Omega,\mathcal{E}^{(T_2)}\Omega\rangle
	&\leqslant CNR^{-\gamma}
		+CM^{1/2}N^{-1/2}(\langle\Omega,T^\ast_2Q_4T_2\Omega\rangle+\ell^{-1})
		+\ell^4\langle\Omega,T^\ast_2Q_4T_2\Omega\rangle
		+C\ell^{1/2} \\
		&\quad+C\ell^{1/2}N^{-1/2}
		+C\langle\Omega,T^\ast_2Q_4T_2\Omega\rangle
		+\varepsilon(\langle\Omega,T^\ast_2Q_4T_2\Omega\rangle+1) \\
		&\quad+\varepsilon^{-1}C(R^{-\gamma}+\ell+N^{-1}+M^{-1}). 
\end{align*}
By choosing $R=N\ell/5$, $M=\ell N$, $\varepsilon=1$, and $\ell=1/4$, we obtain the bound
\[
	\langle T_2\Omega, Q_4T_2\Omega\rangle+\langle\Omega,\mathcal{E}^{(T_2)}\Omega\rangle
	\leqslant C. 
\]

Finally, combining all the above, after optimizing over $\eta > 0$ in \cref{eq:off_diag_HN}, we find that
\begin{align*}
\lambda_1(H_N) 
	&\leqslant \| \ind_+^{\leqslant N}T_1T_cT_2\Omega \|^{-2} \left(\langle \mathcal H\rangle_{T_1T_cT_2\Omega} + C \langle Q_4 \rangle_{T_1T_cT_2\Omega}^{1/2}\,  N^{-(j-1)/2}\right)  \\
	&\leqslant\left(E_N + \langle Q_4\rangle_{T_2\Omega}+\langle\mathcal{E}^{(T_2)}\rangle_\Omega + C \langle Q_4 \rangle_{T_1T_cT_2\Omega}^{1/2}\,  N^{-(j-1)/2}\right) (1+ \mathcal O(N^{-j})) \\
	&\leqslant E_N + C
\end{align*}
for $j=3$. Using the definition of $E_N$, we conclude
\[
	\lambda_1(H_N)\leqslant4\pi\mathfrak{a}N+C,
\]
which implies \cref{eq:gseupperbound}. 

\subsection{Bose--Einstein condensation}\label{sec:bec}
We start by proving \cref{eq:BEC} for every normalized sequence of approximate ground states $\Psi_N$ of $H_N$, i.e., $\langle\Psi_N,H_N\Psi_N\rangle=\lambda_1(H_N)+o(N)=4\pi\mathfrak{a}N+o(N)$. It is easy to check that this is equivalent to proving the following lemma. 
\begin{lem}\label{bec:lem:prelim}
For every normalized sequence of approximate ground states $\Psi_N$ of $H_N$, we have
\begin{equation}\label{eq:BECrdm}
	\lim_{N\to\infty}\frac{\langle u_0,\gamma_{\Psi_N}u_0\rangle}{N}=1,
\end{equation}
where $u_0\equiv1$ and $\gamma_{\Psi_N}$ is the reduced density matrix
\[
	\gamma_{\Psi_N}(x;y)\coloneqq N\int_{\Lambda^{N-1}}\Psi_N(x,z_2,\dotsc,z_N)\overline{\Psi_N(y,z_2,\dotsc,z_N)}\diff z_2\dotso\diff z_N. 
\]
\end{lem}

To prove this result, we follow the strategy in \cite{Lieb_2002} and \cite[Section~5.1]{Lieb_2005}. The main ingredients of the proof are the generalized Poincaré inequality from \cite{Lieb_2002} and a localization of the energy, meaning that the kinetic energy is concentrated in a small subset of the position space, where at least two particles are close to each other.  

Let us recall the generalized Poincaré inequality for $L^2(\Lambda)$, the proof of which can be found in \cite{Lieb_2002}. 
\begin{lem}[generalized Poincaré inequality]\label{bec:lem:pc}
For all subsets $\Omega\subset\Lambda$ and all $f\in H^1(\Lambda)$ with $\int_\Lambda f=0$, we have the inequality
\[
	\int_\Lambda \abs{f(x)}^2\diff x\leqslant C\biggl(\int_\Omega\abs{\gradient f(x)}^2\diff x+\abs*{\Lambda\setminus\Omega}^{2/3}\int_\Lambda\abs{\gradient f(x)}^2\diff x\biggr),
\]
where $C>0$ is some constant. 
\end{lem}

The second ingredient of the proof of \cref{bec:lem:prelim} is the following lemma. 
\begin{lem}\label{bec:lem:loc}
Let $\Psi_N$ be a normalized sequence of approximate ground states of $H_N$ and define
\[
	\Omega_{\mathbf{z}}=\{x\in\Lambda:\min_{k\geqslant2}\abs{x-z_k}>N^{-7/17}\}
\]
for $\mathbf{z}=(z_2,\dotsc,z_N)\in\Lambda^{N-1}$. Then, we have
\begin{equation}\label{bec:lem:loc:eq}
	\lim_{N\to\infty}\int_{\Lambda^{N-1}}\int_{\Omega_{\mathbf{z}}}\abs{\gradient_x\Psi_N(x,\mathbf{z})}^2\diff x\diff\mathbf{z}=0. 
\end{equation}
\end{lem}
With these two lemmas, \cref{bec:lem:prelim} can be shown. The proof is similar to that of \cite[Theorem~1]{Lieb_2002} or \cite[Theorem~5.1]{Lieb_2005} and we skip the details. 

It remains to prove \cref{bec:lem:loc}. 
\begin{proof}[Proof of \cref{bec:lem:loc}]
We want to show that 
\begin{equation}\label{bec:lem:loc:proof:eq1}
\begin{split}
	\MoveEqLeft[3] \frac{\langle\Psi_N,H_N\Psi_N\rangle}{N}-\int_{\Lambda^{N-1}}\int_{\Omega_{\mathbf{z}}}\abs{\gradient_x\Psi_N(x,\mathbf{z})}^2\diff x \\
	&=\int_{\Lambda^{N-1}}\biggl(\int_{\Lambda\setminus\Omega_\mathbf{z}}\abs{\gradient_x\Psi_N(x,\mathbf{z})}^2\diff x
	+\frac{1}{2}\int_{\Lambda}\sum_{k\geqslant2}V_N(x-z_k)\abs{\Psi_N(x,\mathbf{z})}^2\diff x\biggr)\diff\mathbf{z} \\
	&\geqslant 4\pi\mathfrak{a}-o(1). 
\end{split}
\end{equation}
Then, \cref{bec:lem:loc:eq} follows from the assumption $\langle\Psi_N,H_N\Psi_N\rangle=4\pi\mathfrak{a}N+o(N)$. Using the notation
\[
	\Omega_i=\{(x_1,\dotsc,x_N)\in\Lambda^N:\min_{k\neq i}\abs{x_i-x_k}>N^{-7/17}\}
\]
and the fact that $\Psi_N$ is symmetric, we can rewrite \cref{bec:lem:loc:proof:eq1}:
\begin{equation*}
	\sum_{i=1}^N\int_{\Lambda^N\setminus\Omega_i}\abs{\gradient_i\Psi_N}^2+\sum_{1\leqslant i<j\leqslant N}\int_{\Lambda^N}V_N(x_i-x_j)\abs{\Psi_N}^2
	\geqslant 4\pi\mathfrak{a}N-o(N). 
\end{equation*}
It is sufficient to show 
\begin{equation}\label{bec:lem:loc:proof:eq2}
	\sum_{i=1}^N\int_{\Lambda^N\setminus\Omega_i}\abs{\gradient_i\Psi_N}^2+\sum_{1\leqslant i<j\leqslant N}\int_{\Lambda^N}V_{R_0,N}(x_i-x_j)\abs{\Psi_N}^2
	\geqslant 4\pi\mathfrak{a}_{R_0}N-o(N). 
\end{equation}
for fixed $R_0>0$, since this implies
\begin{align*}
	\MoveEqLeft[3] \liminf_{N\to\infty}\frac{1}{N}\biggl(\sum_{i=1}^N\int_{\Lambda^N\setminus\Omega_i}\abs{\gradient_i\Psi_N}^2+\sum_{1\leqslant i<j\leqslant N}\int_{\Lambda^N}V_N(x_i-x_j)\abs{\Psi_N}^2\biggr) \\
	&\geqslant\liminf_{N\to\infty}\frac{1}{N}\biggl(\sum_{i=1}^N\int_{\Lambda^N\setminus\Omega_i}\abs{\gradient_i\Psi_N}^2+\sum_{1\leqslant i<j\leqslant N}\int_{\Lambda^N}V_{R_0,N}(x_i-x_j)\abs{\Psi_N}^2\biggr) \\
	&\geqslant 4\pi\mathfrak{a}_{R_0}\xrightarrow{R_0\to\infty}4\pi\mathfrak{a}. 
\end{align*}
The rest of the proof is similar to the proof of \cite[Lemma~1]{Lieb_2002} or \cite[Lemma~5.2]{Lieb_2005} and we skip the details. 
\end{proof}

\begin{rem}
Note that we did not need \cref{assu:V} \ref{assu:Vc} for this proof, since $\widetilde{V}\in L^1(\R^3)$ already implies $\mathfrak{a}_{R_0}\to\mathfrak{a}$. 
\end{rem}

The proof of the following lemma requires that every normalized sequence of approximate ground states $\Psi_N$, in the sense that $\langle\Psi_N,H_N\Psi_N\rangle=\lambda_1(H_N)+o(N)$, satisfies $\langle\Psi_N,H_N\Psi_N\rangle=4\pi\mathfrak{a}N+o(N)$, which we proved in \cref{sec:gse,sec:gseup}, and that \cref{eq:BEC} holds for every such sequence $\Psi_N$, which we proved in \cref{bec:lem:prelim}. 
It is similar to that of \cite[Lemma~31]{Nam_2023} and \cite[Proposition~20]{Hainzl_2022} and we skip the details. 

\begin{lem}\label{bec:lem:main}
On the truncated Fock space $\mathcal{F}_+^{\leqslant N}$, we have
\[
	UH_NU^\ast=\ind_+^{\leqslant N}\mathcal{H}\ind_+^{\leqslant N}\geqslant 4\pi\mathfrak{a}N+C^{-1}\numop-C,
\]
for some constant $C>0$ independent of $N$. 
\end{lem}

%%%%%%%%%%%%%%%%%%%%%%%%%%%%%%%%%%%%%%%%%%%%%%%%%%%%%%%%
\section{Proof of \texorpdfstring{\cref{thm:mainthm1}}{Theorem 2.2}}\label{ch:proof}

How to obtain optimal BEC \cref{eq:BECrdm_theo} was explained in the previous section; we will now focus on the excitation spectrum. We will use the notations $\mathcal{U}=T_1T_cT_2$, $\widetilde{H}_N=H_N-E_N$ and $\widetilde{\mathcal{H}}=\mathcal{H}-E_N$ with $E_N$ as in \cref{eq:diagonal:RNbound}. The main result in this section is the following lemma.

\begin{lem}\label{ch:proof:lem:main}
Let $0<\kappa<1/13$ such that \cref{eq:Vdecay} is satisfied for 
\[
	\gamma=\frac{1+\kappa}{1-2\kappa}. 
\]
The eigenvalues $\lambda_1(\widetilde{H}_N)\leqslant\lambda_2(\widetilde{H}_N)\leqslant\dots\leqslant\Theta$, with $1\leqslant\Theta\leqslant N^{\kappa/(1-\vartheta)}$, of $\widetilde{H}_N$ on $L^2_\textnormal{s}(\Lambda^N)$ satisfy
\[
	\lambda_L(\widetilde{H}_N)=\lambda_L(\sq(E_\textnormal{Bog}))+\mathcal{O}( N^{-\kappa}\Theta^{1-\vartheta}),
\]
where $0\leqslant\vartheta<1-12\kappa$ is arbitrarily chosen for $\gamma\in[6/5,14/11)$ and 
\[
	\vartheta=\frac{1-2\kappa}{3}
\]
for $\gamma\in(1,6/5)$. 
Here, $\sq(E_\textnormal{Bog})$ is understood as an operator on $\mathcal{F}_+$. 
\end{lem}

\begin{proof}
We denote $R=N\ell/5$, $M_0=\delta_\alpha N$, and $M=\delta_\beta N$. 
Let $1\leqslant\Theta\leqslant N^{\kappa/(1-\vartheta)}$. 
We differentiate between two cases. 
\begin{enumerate}[left=0pt, label=\textit{Case \arabic*:}, ref=case \arabic*]
\item\label{case1} $\gamma\in[6/5,14/11)$. 
We choose $0<\eta\leqslant(1-13\kappa)/2$ and 
\[
	\vartheta=1-\frac{12\kappa}{1-2\eta}\in[1/13,1-12\kappa). 
\]
\item\label{case2} $\gamma\in(1,6/5)$. 
We choose
\[
	\eta=\frac{6-5\gamma}{2\gamma}\quad\text{and}\quad \vartheta=1-\frac{12\kappa}{1-2\eta}=\frac{1-2\kappa}{3}. 
\]
\end{enumerate} 
In both cases, let
\begin{align*}
	\delta_\alpha&=N^{(-2-4\kappa-\eta)/5}\Theta^{-2(1+2\vartheta)/5},
	\qquad \delta_\beta=N^{(-2-4\kappa+4\eta)/5}\Theta^{-2(1+2\vartheta)/5}, \\
	\ell&=N^{-2\kappa}\Theta^{-2\vartheta},
	\qquad\varepsilon=N^{-\kappa}\Theta^{-\vartheta},
\end{align*}
and $j\geqslant1$ such that
\[
	\eta\geqslant\frac{3\kappa}{j}. 
\]

With this choice of parameters, we can conclude the proof using the strategy used for \cite[Lemma~32]{Nam_2023} and \cite[Theorem~1]{Hainzl_2022}, which consists of first proving the lower bound on $\lambda_L(\widetilde{H}_N)$ and then proving the matching upper bound on $\lambda_L(\widetilde{H}_N)$. 
We skip the details of the proof. 
\end{proof}

To conclude the proof of \cref{thm:mainthm1}, we note that $\lambda_1(\sq(E_\textnormal{Bog}))=0$. Hence, by \cref{ch:proof:lem:main}, we have $\lambda_1(H_N)=E_N+\lambda_1(\widetilde{H}_N)=E_N+\mathcal{O}(N^{-\kappa})$ for any admissible $0\leqslant\kappa<1/13$, which shows the first part of \cref{thm:mainthm1}. This, in combination with \cref{ch:proof:lem:main}, also implies the second part of the theorem, since the spectrum of $\widetilde{H}_N$ below a threshold $\Theta\leqslant N^{\kappa/(1-\vartheta)}$ consists of eigenvalues of the form 
\[
	\sum_{p\in 2\pi\Z^3\setminus\{0\}}n_p\sqrt{p^4+16\pi\mathfrak{a}p^2}+\mathcal{O}(N^{-\kappa}\Theta^{1-\vartheta})
\]
with $n_p\in\N_0$ for all $p\in2\pi\Z^3\setminus\{0\}$, where only finitely many $n_p$ are non-zero, and we have 
\[
	\lambda_L(H_N-\lambda_1(H_N))=\lambda_L(\widetilde{H}_N)+\mathcal{O}(N^{-\kappa})
\]
for the $L$-th lowest eigenvalue $\lambda_L(H_N-\lambda_1(H_N))$ of $H_N-\lambda_1(H_N)$ below $\Theta$. 

\subsection*{Acknowledgements}

LR thanks Lea Boßmann for helpful feedback on the manuscript. This work was partially funded by the Deutsche Forschungsgemeinschaft (DFG, German Research Foundation) – Project-ID 470903074 – TRR 352.

\emergencystretch=1em
\printbibliography[title={References},heading=bibintoc]

\end{document}